\documentclass[margin=0.5in]{article}
\usepackage{bm}
\usepackage{xcolor}

\usepackage{amsmath}      % For advanced math formatting (e.g., align, cases)
\usepackage{amssymb}      % For additional math symbols
\usepackage{amsthm}       % For theorem, lemma, and remark environments
\usepackage{algorithm}    % For the algorithm environment
\usepackage{algpseudocode} % For algorithmicx pseudocode style
\usepackage{graphicx}
\usepackage{hyperref}
\usepackage{comment}
\usepackage{cleveref}
\usepackage{mathtools}    % For enhanced math features (e.g., \coloneqq)
 \usepackage{thmtools}
\usepackage{placeins} % in your preamble

\usepackage{pgfplotstable} 
\usepackage{filecontents}
\usepackage{tikz}
\usepackage{pgfplots}
\usepackage{subcaption} % Add this in your preamble
\usetikzlibrary{patterns}
\usetikzlibrary{intersections}
\pgfplotsset{compat=1.17}
\usepackage{pgf-pie}

\newcommand{\V}{V}

\theoremstyle{definition}
\newtheorem{example}{Example}[section]
\newtheorem{theorem}{Theorem}[section]
\newtheorem{lemma}{Lemma}[section]
\newtheorem{proposition}{Proposition}[section]
\newtheorem{corollary}{Corollary}[section]
\newtheorem{remark}{Remark}[section]
\newtheorem{definition}{Definition}[section]

\title{Strategic Analysis of Just-In-Time Liquidity Provision in Concentrated Liquidity Market Makers}
\date{}

\author{Bruno {Llacer Trotti}\footnote{brunolt@ic.ufrj.br} \\UFRJ
\and
Weizhao {Tang}\footnote{twz17519@gmail.com } \\ 
CMU
\and
Rachid {El-Azouzi}\footnote{rachid.elazouzi@univ-avignon.fr} \\
Avignon University
\and
Giulia Fanti\footnote{gfanti@andrew.cmu.edu} \\
CMU
\and
Daniel Sadoc {Menasché}\footnote{sadoc@dcc.ufrj.br} \\
UFRJ
}

\newtheorem*{reptheorem}{Theorem}
\newtheorem*{replemma}{Lemma}
\newtheorem*{repcorollary}{Corollary}
\newtheorem*{repproposition}{Proposition}

\newif\ifshowcomments
\ifshowcomments
  \newcommand{\dsm}[1]{\textcolor{red!50}{DSM: #1}}
  \newcommand{\gf}[1]{\textcolor{blue}{GF: #1}}
  \newcommand{\re}[1]{\textcolor{green}{RE: #1}}
  \newcommand{\bt}[1]{\textcolor{violet!50}{BT: #1}}
  
  \newcommand{\wt}[1]{\textcolor{orange}{[Weizhao: #1]}}
\else
  \newcommand{\dsm}[1]{}
  \newcommand{\gf}[1]{}
  \newcommand{\re}[1]{}
  \newcommand{\bt}[1]{}
  
  \newcommand{\wt}[1]{}
\fi

\usepackage{cite}

\begin{document}

\maketitle

%TODO mandatory: add short abstract of the document
\begin{abstract}
Liquidity providers (LPs) are essential figures in the operation of automated market makers (AMMs); in exchange for transaction fees, LPs lend the liquidity that allows AMMs to operate.
While many prior works have studied the incentive structures of LPs in general, we currently lack a principled understanding of a special class of LPs known as Just-In-Time (JIT) LPs. These are strategic agents who momentarily supply liquidity for a single swap, in an attempt to extract disproportionately high fees relative to the remaining passive LPs.
% Concentrated Liquidity Market Makers (CLMMs), such as Uniswap V3, allow liquidity providers (LPs) to allocate capital over specified price ranges, improving capital efficiency but also enabling new forms of strategic behavior. Among these, Just-in-Time (JIT) liquidity provision has emerged as a key mechanism, whereby agents momentarily supply liquidity for a single swap, seeking to extract fees while minimizing exposure. 
This paper provides the first formal, transaction-level model of JIT liquidity provision for a widespread class of AMMs known as Concentrated Liquidity Market Makers (CLMMs), as seen in Uniswap V3, for instance.
% in CLMMs. 
We characterize the landscape of price impact and fee allocation in these systems, formulate and analyze a non-linear optimization problem faced by JIT LPs, and prove the existence 
% and uniqueness
% \gf{do we prove uniqueness?}\bt{No we do not, this is wrong} 
of an optimal strategy. 
% Our model accounts for trade size, liquidity depth, and pool parameters, and yields actionable conditions under which JIT provision is profitable.  
By fitting our  optimal solution for JIT LPs to real-world CLMMs, we observe that in liquidity pools (particularly those with risky assets), there is  a significant gap between observed and optimal JIT behavior. Existing JIT LPs often fail to account for price impact; doing so, we estimate they could increase earnings by up to 69\% on average over small time windows. 
% more than what they currently do
% \gf{what do you mean by "they can earn 69\% more"? More than they currently earn?\bt{Yes} In the intro, it says 50\% suboptimality, can we use consistent numbers, and explain it precisely?}\bt{Done} 
% over small time windows. 
We also show that JIT liquidity, when deployed strategically, can improve market efficiency by reducing slippage for traders,  
 albeit at the cost of eroding average passive LP profits by up to 44\% per trade. 
 % Our findings bridge the gap between theoretical modeling and empirical behavior of liquidity provision in modern AMMs.
\end{abstract}

% https://docs.google.com/document/d/19ELfbp-WPk2MsUsXGL8GawJM45AehKXVjnigluAVAQ0/edit?usp=sharing

\section{Introduction}
\label{sec:typesetting-summary}

Recent years have seen increased adoption of Decentralized Exchanges (DEXs)~\cite{yukselretrospective} and  Automated Market Maker (AMM) protocols \cite{DeFi}, which automatically set trade prices and execute trades;
% .  Currently, several widely-used AMM smart contracts operate on blockchains as token exchanges 
prominent markets running AMMs include Uniswap~\cite{UniswapV3Core}, SushiSwap~\cite{sushiswap2020},  and PancakeSwap~\cite{pancakeswap2021}. At the time of writing, the annual trading volume on AMM-based exchanges exceeds hundreds of billions of dollars~\cite{nasdaq2024pancakeswap}.

Liquidity providers (LPs) are users who lend tokens to specific pools, thereby providing liquidity for trades in exchange for trading fees. LPs are central to the effective operation of AMMs; consequently, many prior works have sought to analyze their behavior~\cite{miori2023defi,berg2022empirical,angeris2021analysisuniswapmarkets,cartea2024decentralized,tang2024game}. For example, several of these works characterize conditions under which LPs are incentivized to contribute liquidity, and strategies for doing so to optimize returns \cite{tang2024game, Reward}.% However, \gf{these prior works make simplifying assumptions that significantly limit our theoretical understanding of CLMMs:} 

Despite a growing body of work analyzing the incentives of LPs in AMMs, to our knowledge, few have studied an  important class of LPs known as \emph{Just-in-Time Liquidity Providers} (JIT LPs)~\cite{JIT,JITparadox}.  These are informed   strategic agents who aim to extract profit by providing liquidity in a highly targeted manner. Rather than holding a liquidity position over a period of time, a JIT LP   provides liquidity for a single trade. 
This can be executed by a sandwich attack \cite{sandwich-attacks} on the \emph{target} transaction
%, also referred to as 
by executing a swap that involves frontrunning the target transaction with an ``add liquidity''  transaction and backrunning it with a ``withdraw liquidity''  transaction.  
By temporarily capturing a higher share of liquidity, the JIT trader can earn a larger portion of the trading fee that is paid to all LPs pro rata to each LP's liquidity share.\footnote{Note that the ``trading fee'' (pool fee) is a percentage of the swap value paid by the trader to liquidity providers, whereas the ``transaction fee'' (gas fee) is paid to Ethereum validators to include the transaction in a block. In this paper, except otherwise noted, fees refer to trading fees.}

Although JIT LPs may appear to operate at the margins of an AMM, they play a crucial role in the ecosystem. They act as a frictional force against price movements, they help pools stay aligned with external market opportunities, and they reduce slippage for traders \cite{cryptoeprint:2023/973, JIT }.\footnote{Slippage is a phenomenon by which traders lose money in an AMM due to price changes over the course of a single trade.} 
%\gf{do we have a citation?}\bt{I dont, this is actually something that we show later on this paper}
At the same time, they introduce adversarial costs to passive LPs and arbitrageurs. 
% Despite being underexplored in theory, 
The prevalence of JIT LPs %has increased over time \ct, 
% mainly within CLMMs that offer better use of resources,  and
 is expected to grow % further
 with the emergence of technologies such as Flashbots \cite{Flashbots} and Uniswap V4 \cite{UniswapV4}, which introduce optimized mechanisms for JIT operations (e.g., hooks, bundled transactions). However, we currently lack a theoretical understanding of 
 % possibilities offered by these innovations  being overlooked, and 
 JIT liquidity provision and its effects on AMMs. 
 % Without sufficient theoretical explorations, the opportunities offered by these innovations may be left unrealized and JIT LPs may still be overlooked.

This work fills a current gap in the literature by developing and analyzing a principled model of JIT LP incentives for an  important class of AMMs known as \emph{Concentrated Liquidity Market Makers (CLMMs)}. %\gf{
Introduced in Uniswap~V3~\cite{UniswapV3Core}, CLMMs allow LPs to invest liquidity over fine-grained price ranges; the LP's liquidity is only used to support trades when the AMM price lies within the specified range. This helps LPs reduce risk due to large price fluctuations.
%} \gf{
Today, CLMMs are widely used in Uniswap V3 and the mechanism is generalized in emerging Uniswap V4 markets \cite{UniswapV4}. 

%\gf{
% Under the CLMM paradigm, we formulate and analyze a
%}
% model that reflects the core mechanics of real-world CLMMs, including fine-grained liquidity allocation, virtual reserves, and pro rata trading fee distribution (\S\ref{sec:backgr}). We further characterize the global structure of price impact, 
Under a model of the CLMM paradigm, we formulate and analyze the JIT resource allocation optimization problem, study its solution properties analytically, and provide an algorithm for computing optimal strategies (\S\ref{sec:algo}). Our analysis reveals new insights into how optimal JIT strategies depend on transactional and pool-specific parameters. 
% We also investigate on-chain data and show that existing JIT strategies often neglect price impact, leading to suboptimal outcomes. This highlights both the unrealized potential of JIT agents and the important role they could play in shaping AMM dynamics (\S\ref{sec:data-analysis}). 

% Finally, we discuss the consequences of our findings and propose new research opportunities regarding better understanding and possible improvements on CLMMs (\S\ref{sec:concl}).

% This work aims to address existing gaps by providing a better understanding and modeling of optimal behaviors in AMMs and JIT LP strategies. While doing so, we provide the following contributions. 

%By the end, we will have made the following contributions:

\begin{comment}
    
\begin{itemize}
\item Demonstrated the the behavior and global ladscape of the price impact function.
\item Formalized the JIT LP optimization problem.
\item Proved the existence of a solution and proposed an algorithm for finding it.
\item Fully characterized the optimal solution in terms of pool parameters.
\item Analyzed empirical data from real AMMs, identifying the underperformance of current JIT strategies.
\item Proposed a stronger theoretical framework and intuition for understanding the dynamics of AMM systems and their agents.
\end{itemize}

\end{comment}

%\subsection*{Contributions}
We summarize our contributions as follows.
\begin{itemize}
    \item \textbf{Price impact landscape:} 
    We analyze how swap-induced price changes affect LP revenue and formally characterize the conditions under which liquidity provision leads to gains or losses. 
    These effects are often called ``impermanent loss''; we instead use the term ``price impact'' to reflect that they do not always result in losses (Def. \ref{def:absprim}).
    This analysis yields a precise description of the parametric region in which LPs benefit from price movements  (\S\ref{sec:imperloss}). This result holds for both \emph{passive} (regular) LPs and JIT LPs.

    \item \textbf{Formulation and analysis of utility optimization for JIT LPs:}
    We model the decision-making process of JIT LPs as a  non-linear optimization problem and define a transaction-level utility function that captures the trade-off between fee revenue and price impact. This formulation enables a rigorous exploration of the strategy and utility space  (\S\ref{sec:fees} and \S\ref{sec:jit-optimization}) . We classify optimal JIT behavior across three  swap scenarios—overpriced, arbitrage-driven, and overshooting trades—and present an algorithm for identifying optimal liquidity positions  (\S\ref{sec:algo} and~\S\ref{sec:data-analysis}).

    \item \textbf{Empirical insights on real-world JIT performance:}
    Using on-chain data, we report three key empirical findings (\S\ref{sec:data-analysis}). First, JIT traders today appear to allocate trades suboptimally, without accounting for price impact; our experiments suggest that this unawareness reduces JIT profit by up to 41\% on average.
    % \gf{make this percent match abstract}\bt{Done, they could gain up to 69\% than they currently doing, so they are reducing their potential profits in 40.89\%. Is better if we speak in the same words ?}
    % than what is possible.} 
    %leading many JIT LPs to miss profitable opportunities or incur losses. 
    Second, although JIT LPs currently account for a small share of total fee revenue, 
    %their primary value comes from capturing 
    most of JIT revenue comes from price impact (93.7\%) rather than extracting fees from passive LPs. Third, JIT liquidity provision can become a serious competitor to passive LPs by reducing their fee income by up to 40\%, while slightly reducing the cost of slippage for traders. 
    %When properly implemented, JIT strategies can enhance overall market efficiency and serve as a net positive force for traders in decentralized markets.
\end{itemize}

% \textbf{Outline. } The remainder of this paper is organized as follows.  In \S\ref{sec:backgr} we briefly introduce the required background. Then, \S\ref{sec:jit} reports our analytical results, followed by empirical  results using on-chain data in \S\ref{sec:data-analysis}. \S\ref{sec:related} presents related work and \S\ref{sec:concl} concludes. 

\section{Background and Model}
\label{sec:backgr}

\subsection{Introduction to AMMs}

Automated Market Makers (AMMs) enable decentralized token exchanges through algorithmic pricing instead of traditional order books. 
We begin by explaining one of the most common classes of AMMs known as constant product market makers (CPMMs) \cite{CONSTANT}. 
A CPMM maintains a liquidity pool of tokenized assets $X$ and $Y$.
After initialization, the CPMM maintains a pool price $q$. Throughout this paper, \(q\) denotes the amount of token \(Y\) per unit of token \(X\), i.e., the price of token \(X\) denominated in token \(Y\). Equivalently, \(1/q\) is the price of token \(Y\) denominated in token \(X\). 
When an LP enters the CPMM, they contribute an amount of liquidity $L$, which comprises token amounts $x$ of $X$ and $y$ of $Y$, according to the following rule: 
% (denoted by , respectively)
%
\begin{equation}
    x = \frac{L}{\sqrt{q}}, \qquad y = L \sqrt{q}. 
    \label{eq:pos-amount-legacy}
\end{equation}

Consequently, the token amounts $x, y$ of each liquidity position of amount $L$ always follow the constant product rule $xy = L^2$. 
Additionally, when there exist multiple liquidity positions $\{(x_i, y_i, L_i)\}_{i=1}^N$, the corresponding summations also follow the constant product rule 
\begin{equation}
    x_{\text{total}} \cdot y_{\text{total}} = \textstyle \left(\sum_{i=1}^N x_i \right) \cdot \left(\sum_{i=1}^N y_i \right) = \left(\sum_{i=1}^N L_i \right)^2 = L_{\text{total}}^2. \label{constForm}
\end{equation}

A CPMM supports three basic operations -- mint (add liquidity), burn (remove liquidity), and swap (trade). 
\begin{itemize}
    \item Mint and burn transactions are executed by an LP who specifies liquidity $L$ to add/remove, and will deposit/withdraw token amounts $(x, y)$ following \eqref{eq:pos-amount-legacy}. $q$ remains constant. 

    \item Swaps are executed by a trader who specifies an amount of either token to pay. Without loss of generality, assume the trader pays $\Delta x$ tokens of type $X$. 
    $L$ remains unchanged for each liquidity position, so the trader will get $\Delta y$ $Y$ tokens following \eqref{constForm}:
    \[
        (x_{\text{total}} + \Delta x) (y_{\text{total}} - \Delta y) = L_{\text{total}}^2. 
    \]
    As a result, the pool has a new price $q' = (y_{\text{total}} - \Delta y) / (x_{\text{total}} + \Delta x)$ and the equivalent token amount $(x', y')$ of each liquidity position $L$ will update by plugging $q'$ into \eqref{eq:pos-amount-legacy}. 
\end{itemize}

% \begin{definition}[Slippage]
% \label{def:slippage}
\emph{Slippage} is an important concept in this domain; it refers to the difference between the execution price of a trade and the price expected at the time of trade initiation. 
% \end{definition}
In the context of AMMs, slippage arises from the fact that large trades move the price along the bonding curve, resulting in a marginally worse exchange rate for each successive unit traded. Slippage is an inherent byproduct of price impact in   AMMs.

\begin{example}[Slippage]
\label{ex:slippage}
% \gf{There's currently no example 1} 
Consider an AMM with 100 ETH and 10,000 USDC, so \( q = \frac{\textbf{\#USDC}}{\textbf{\#ETH}} = 100 \), i.e., the pool price is 100 USDC per ETH. Suppose a trader wishes to buy 10 ETH. Due to the constant-product formula \( x \cdot y = k \), the required USDC input is such that 
$
100 \times 10{,}000 = (100 - 10) \times y_{\text{final}}$, i.e., $y_{\text{final}}  \approx 11{,}111.11$.
Thus, the trader must pay approximately 1,111.11 USDC to receive 10 ETH, implying an average price of 111.11 USDC per ETH — higher than the original pool price. The slippage is:
$
({111.11 - 100})/{100} = 11.11\%$.
This increase reflects the adverse price movement caused by the trader’s own order.
\end{example}

\subsection{Concentrated Liquidity Market Makers}

Since the emergence of CPMMs, more complex AMM designs have been proposed; among the most widespread of these is  Concentrated Liquidity Market Makers (CLMMs) \cite{CLMM}. They allow LPs to choose not only the amount of liquidity they wish to add to the pool, but also the price range in which the liquidity is active.
% This design introduces a critical innovation known as \emph{liquidity segmentation}, where capital is allocated more efficiently and liquidity is deeper within specified price ranges. Specifically, w
When an LP adds liquidity to the CLMM, the LP  creates a \emph{liquidity position} by specifying a tuple $(L, a, b)$, where $L$ is the amount of liquidity, and $(a, b)$ represents the price range in which the position is effective. During a trade, if the AMM price exits the specified range $(a,b)$ of an LP's position, that liquidity will not be used to support the trade outside the specified price range. At price $q$, this liquidity position reserves token amounts 
\begin{equation}
x = x^{(a,b)}(q) \triangleq L \cdot \left(\frac{1}{\sqrt{\hat q_{a,b}}}-\frac{1}{\sqrt b}\right), \qquad
y = y^{(a,b)}(q) \triangleq L \cdot \left(\sqrt{\hat q_{a,b}}-\sqrt a\right),    \label{eq:position-holdings}
\end{equation}
where \(\hat q_{a,b}=\min\{b,\max\{a,q\}\}\) is the projection of the price $q$ onto the interval $(a,b)$.  By setting $(a, b) = (0, \infty)$, this relation reduces to \eqref{eq:pos-amount-legacy}, which implies that a CPMM is a special case of a CLMM when every LP chooses the full price range. However, in general cases where $0 < a < b < \infty$ , the reserves of the CLMM liquidity position remain    constant with $x = 0$ when $q \ge b$ and with $y = 0$ when $q \le a$. This implies that 1) the position is unused when $q$ moves out of range $(a, b)$, and 2) there exists a maximum reserve of token $X$ and token $Y$ for such positions. 

A CLMM also supports mint, burn and swap operations, as long as they are consistent with the total reserves in the pool (e.g. do not try to swap tokens that do not exist). 
% with some key differences:
% % \begin{itemize}
%     (1) When an LP executes a burn, the LP can only choose a tuple $(L, a, b)$ if it   already owns a position $(L', a, b)$ with $L' \ge L$. 
%     (2) When a trade executes a swap by inputting $\Delta x$, the operation will be rejected if $x_{\text{total}} + \Delta x$ exceeds the sum of maximum $X$ token reserves over all liquidity positions. 
    % As long as theseOtherwise, 
    A swap outputs $\Delta y$ in the following steps:
  \begin{enumerate}
 \item \textbf{Calculate the new price $q'$}: Since, for each position, $x$ is a continuous and monotonic function of $q$, the total liquidity reserve $x_{\text{total}}$ is also continuous and monotonic in $q$. 
Therefore, we can determine the post-swap price $q'$ by solving for the value of $q'$ that satisfies $x_{\text{total}}(q') = x_{\text{total}}(q) + \Delta x(q')$, which is well-defined due to the invertibility of $x_{\text{total}}(\cdot)$. 
Provided that $q'$ remains within the union of all active price ranges, the swap can be successfully executed (refer to~\cite{UniswapV3Core} for implementation details).
    \item \textbf{Update the token reserve}: Each position keeps $(L, a, b)$ constant, so we recalculate $x'$ and $y'$ by plugging $q'$ into~\eqref{eq:position-holdings}. The change in the second token, $\Delta y$, can then be obtained by summing all $y'$ values.
\end{enumerate}
% \end{itemize}

% Compared to a CPMM, a CLMM allows LPs to provide the same amount of liquidity $L$ within a limited range at a cheaper cost of tokens. Alternatively, by paying the same cost within a limited price range, an LP provides a higher amount of liquidity. %As a result, when the pool price movement is within the range, the position \dsm{which position?} can support a higher volume of trades. 
% at higher volume with the same price} movement. 

By concentrating liquidity within a narrower price range, CLMMs enable each active position to support a higher volume of trades with the same amount of capital. This increased capital efficiency was designed to reduce price slippage for traders - as larger trades can be absorbed with smaller price movements \cite{UniswapV3Core}- and to allow LPs to extract more from their money, since they can have a higher liquidity, therefore a higher share of the fees, using the same budget by picking a narrower range.
At the same time, when the market price exits a position's specified range, that liquidity becomes inactive or ``frozen,'' and only the remaining active positions support the trade. As a result, traders of such trades may experience higher slippage compared to a scenario where liquidity is distributed across the entire price spectrum.

In practice, CLMMs discretize the price axis using a system of \emph{ticks}. Let $T$ be the set of possible ticks,  $T = \{ t_i \;|\; i = 0, \ldots, M \}$,   with $M \in \mathbb{Z}_{>0}$.\footnote{Ticks are defined by exponential spacing: $T = \{ t_i \;|\; i = 0, \ldots, M \}$, where $t_i = 1.0001^{\tau (i - \iota)}$ with $\tau, M \in \mathbb{Z}_{>0}$ and $\iota \in \mathbb{Z}$.} As a result, the set of valid price ranges is given by $\mathcal{R} = \{(a, b) \in T \times T \mid a < b\}$. Liquidity providers are restricted to choosing ranges $(a, b) \in \mathcal{R}$, so liquidity remains constant between consecutive ticks $(t_m, t_{m+1})$. Within each such interval (or across consecutive intervals with identical liquidity), price changes follow the structure defined in~\eqref{eq:position-holdings}, which reveals a tight coupling between price movement and the liquidity available in the pool.

\subsubsection{Fees} \label{sec:fees}

One of the incentives for liquidity providers to participate in the pool is the collection of fees. Each transaction pays a fee proportional to the input amount, and this fee is distributed pro-rata among all providers with active liquidity in the corresponding price range. The fee amount varies across pools, each characterized by a parameter \( \alpha \in (0,1) \) that determines the fraction of the input amount collected as a trading fee.

In a CLMM, every time a provider mints a position \( (L,a,b) \), they increase the liquidity present in every tick \(\ t_m \in (a,b)\). To assess the liquidity available at a specific tick \( m \), we  sum the liquidity of all positions that include it. Formally, we define the  liquidity provided at tick \( m \) 
% \gf{do you mean $k$?}\bt{we are now using $m$ for tick} 
by a provider \( n \in [N]\) as \( P_{n,m}\) and calculate the total liquidity as: 
\begin{equation}
    P_m = \sum_{n \in [N]} P_{n,m}
\end{equation}
Let \( \delta = (\delta_m)_{m=1}^{M} \) 
 denote the vector of fees (in dollars) collected per tick during a swap.
Then,  \( \delta_m \triangleq \Delta  x_m {\cdot}  \alpha  {\cdot}p_x\), where \(\Delta x_m\) is the amount of tokens exchanged in tick $m$ 
% (see Appendix~\ref{sec:prooflasttheo}) \gf{broken link}:
\begin{equation*}
 \Delta x_m =
    P_{n,m}\left(\frac{1}{\sqrt{\hat{q}_m'}} - \frac{1}{\sqrt{a_m}}\right).
\end{equation*}

As discussed above, the total fee from  a swap is distributed pro-rata among all providers with active liquidity. The fee earned by provider \( n \) is given by:
\[
\mathcal{F}_n = \sum_{m=1}^{M} \delta_m \cdot \frac{P_{n,m}}{P_m},
\]
where \( \delta_m \) is the total fee from swap \( m \) and \( P_{n,m} \) is the liquidity contributed by passive provider \( n \) over tick \( m \).

\subsection{Price Impact (a.k.a. Impermanent Loss)} \label{sec:imperloss}

Prior literature commonly assumes that the market-implied price ratio between two tokens is stable and externally set, given by \(q_{\textrm{market}} = \frac{p_x}{p_y}\), where \( p_x \) and \( p_y \) denote the market prices of tokens \( X \) and \( Y \), respectively \cite{tang2024game,LVRvsILI,LVRvsILII}. Under our convention, this ratio is the amount of token \(Y\) per unit of token \(X\), i.e., the market price of token \(X\) denominated in token \(Y\). 
% \gf{add 1 more citation?}\bt{Done}
% This assumption underlies much prior work—including \cite{tang2024game}—which models AMM environments under the premise that market prices are exogenous and stable throughout the analysis.
However, the actual price within a CLMM pool is determined endogenously—it is only changed by trades in the pool. 
If it does not coincide with the market price, 
arbitrage activities become profitable. 
% and \emph{impermanent loss} emerges for LPs. 
At equilibrium—when no arbitrage opportunities remain—the two values align, such that \( q = q_{\text{market}} \). This alignment is often assumed in theoretical analyses, particularly in the absence of active participants like JIT LPs and arbitrageurs. 
% who benefit from transient price discrepancies.

The divergence between the pool price \(q\) and the market price thus represents both a risk and an opportunity. On one hand, price mismatches expose LPs to impermanent loss—the difference in value between a liquidity position and the equivalent token holdings outside the pool. On the other hand, such mismatches create arbitrage opportunities for sophisticated agents, including JIT LPs.  To encompass both favorable and adverse outcomes of price divergence, we adopt a more general term and refer to the impermanent loss effect as \emph{price impact} throughout the remainder of this paper.

Consider an LP $n$ who mints a position when the pool price is \( q \) and later withdraws it at pool price \( q' \). Let \( p_x \) and \( p_y \) denote the external market prices of tokens \( X \) and \( Y \) at the time of minting, and let \( p_x' \) and \( p_y' \) be their respective market prices at the time of withdrawal.  Following~\eqref{eq:position-holdings}, we denote by \( x_n(q) \) and \( y_n(q) \) the token amounts associated with  the position of LP~\( n \)  when the pool price is \( q \) (at minting), and by \( x_n(q') \) and \( y_n(q') \) the corresponding amounts when the price is \( q' \) (at withdrawal).
% }  \gf{Virtual holdings are not defined, notation is diff from (3)} \dsm{is it good now? do we need to explicitly define virtual holding?} 
The dollar value of the position at the time of minting and at the time of withdrawal is:
\begin{align*}
V_\text{mint}(L,  p_x, p_y, q) &= p_x \cdot x_n(q) + p_y \cdot y_n(q), \\
V_\text{withdraw}(L, p'_x, p'_y, q') &= p_x' \cdot x_n(q') + p_y' \cdot y_n(q').
\end{align*}

\begin{definition}[Absolute price impact] \label{def:absprim}
The \emph{absolute price impact} is the net change in value of the position:
\begin{equation}
\mathcal{C}_n(\Delta x, L; p_x, p_y, p'_x, p'_y, q, q') \triangleq V_\text{mint}(L,  p_x, p_y, q) - V_\text{withdraw}(L,  p'_x, p'_y, q'). \label{eq:priceimpact}
\end{equation}
\end{definition}

In what follows, to simplify notation we omit the explicit dependencies on parameters of the above quantities, e.g., referring to  $\mathcal{C}_n(\Delta x, L; p_x, p_y, p'_x, p'_y, q, q') $  simply as $\mathcal{C}_n$. 
Note that a \emph{negative} price impact (\( \mathcal{C}_n < 0 \)) corresponds to a \emph{gain} for the liquidity provider relative to holding the assets outside the pool. Conversely, a \emph{positive} price impact implies a relative loss.

\begin{definition}[Relative price impact]
The   \emph{relative price impact} is the ratio of this value difference to the initial position value:
\begin{equation}\label{eq:RelativePrice}
\mathbf{PI} \triangleq {\mathcal{C}_n }/{V_\text{mint} }. 
\end{equation}
\end{definition}

%, a negative value of \( \mathbf{PI} \) indicates that the LP’s position performed better than simply holding the tokens outside the pool.
Throughout the rest of this paper, we omit the LP subscript \( n \) when the context is clear, and refer to quantities such as \( x_n(q) \), \( y_n(q) \) simply as \( x(q) \), \( y(q) \).

\section{When is Price Impact Beneficial to LPs?}
\label{sec:price-impact-conditions}

As explained previously, an LP $n$ gains when the fees accrued $\mathcal F_n$ are larger than the price impact $\mathcal C_n$. 
However, LPs may struggle to predict  these tradeoffs  for three reasons:
(1) Price fluctuations and transaction patterns are inherently unpredictable. 
(2) Given information about price and transaction fluctuations, we currently lack closed-form expressions showing when LPs will benefit as a function of transaction size and current market conditions. 
(3) Even if we had knowledge of items (1) and (2), most LPs invest on a slow timescale (days or even weeks \cite{liqFragmentation}), over which it is not possible to take advantage of per-transaction profits.

% Moreover, they can sometimes circumvent restriction (1) by . 
A key observation is that Just-In-Time (JIT) LPs do not suffer from the 1st and 3rd constraints---they can act quickly enough to move on a per-transaction basis using current price data. 
Hence, if we can resolve item (2), JIT LPs can  determine with some certainty when CLMM transactions will be profitable. 
% under  (possibly partial) information about transaction size and  value. 

In this section, we resolve challenge (2) by formally characterizing the conditions under which \( \mathcal{C} \) is non-positive or positive in CLMMs; we show that the answer depends on whether the pool price moves toward or away from the prevailing market price. 
Our results in this section are not specific to JIT LPs, but as mentioned previously, passive LPs may not be able to take advantage of these results, as they move their liquidity on much slower time scales.

\noindent \textbf{Assumptions.~~}
While deriving the following results,  we assume that the market prices of tokens \( X \) and \( Y \) remain constant throughout the trade, i.e., \( p_x = p_x' \) and \( p_y = p_y' \). This is because the trade occurs instantaneously when the block including the transaction is finalized on the blockchain. 
In this case, the price impact simplifies to the difference in the dollar value of the LP’s position before and after the trade, evaluated using fixed market prices. According to Definition~\ref{def:absprim}:
% \gf{I suggest we put the price impact definition in a formal definition environment to make it easy to find.}
\begin{align}
\mathcal{C}(\Delta x, L; p_x, p_y, q, q')
&= p_x  \left(x(q) - x(q')\right) + p_y  \left(y(q) - y(q')\right) \nonumber \\
&= -p_x  \Delta x(q, q') -p_y  \Delta y(q, q'). \label{eq:priceimpact1b}
\end{align} 
In what follows, to simplify notation we refer to $\mathcal{C}(\Delta x, L; p_x, p_y, q, q')$ simply as $\mathcal{C}$.
% \gf{I think the notation $\mathcal C$ should depend on several variables, like $q$ and $q'$, $p_x$, $p_y$ (or maybe $\Delta x$ and $L$?). We might want to define it that way (i.e. with dependency on other variables it depends on), and then say for simplicity of  notation, we will drop the dependency on other variables going forward. That way when people refer to this definition, they will at least see which variables $\mathcal C$ depends on.}

% We begin by considering the case in which the pool price is initially aligned with the external market price. In this setting, the liquidity provider is guaranteed not to incur losses from price movement alone. This is formalized in the following theorem. The proof, along with those of the subsequent results, is provided in Appendix~\ref{sec:proofsectionbasicimpact}.

% \dsm{fix flow here}
% Next, we provide 
\subsection{Conditions for Favorable Price Impact}
Our first main result consists of necessary and sufficient conditions under which a price change \( q \to q' \) results in non-positive price impact. This threshold condition delineates precisely when an LP's position becomes favorable or unfavorable relative to simply holding the tokens at market prices. 
% (proofs in Appendix~\ref{app:price-impact}) \gf{broken link}:

\newcounter{tempEqCnt} % Temporary counter for tracking equation number
\refstepcounter{equation}
\setcounter{tempEqCnt}{\value{equation}}

\begin{theorem}[Threshold condition for price impact]\label{theorem:price_impact} 
Let       $\hat{q} \triangleq \max\{a, \min\{b, q\}\}$, and similarly for $\hat{q}'$. 
% \gf{I would move this definition before the bulleted list.} 
The change in pool price relates to price impact as follows:
\begin{equation*}
\begin{array}{c|c}
\begin{minipage}{0.45\textwidth}
\textbf{Case 1:} \( q' < q \): \( \mathcal{C} \le 0 \) if and only if
\begin{equation}
\frac{1}{\hat{q}} \left( \frac{p_x}{p_y} \right)^2 \ge \hat{q}'.\tag{\thetempEqCnt}  \label{ineq1}
\end{equation}
\end{minipage}
&
\begin{minipage}{0.45\textwidth}
\textbf{Case 2:} \( q < q' \): \( \mathcal{C} \le 0 \) if and only if
\begin{equation}
\frac{1}{\hat{q}} \left( \frac{p_x}{p_y} \right)^2 \le \hat{q}'. 
\tag{\number\numexpr\value{tempEqCnt}+1\relax} \label{ineq2}
\end{equation}
\end{minipage}
\end{array}
\end{equation*}
\setcounter{equation}{\numexpr\value{tempEqCnt}+1\relax}
\end{theorem}
\begin{proof}[Proof sketch]
(Full proof in Appendix \ref{sec:priceimpactproof})
We establish    conditions  under which an LP experiences non-positive price impact. Substituting \(p_x = p_x', p_y = p_y'\) into~Eq.\eqref{eq:RelativePrice}:
\begin{equation}
 \mathbf{PI} = 1-\frac{V_\textrm{withdraw}(L,p_x,p_y,q')}{V_\textrm{mint}(L,p_x,p_y,q)}  \le 0 \iff 1 \le \frac{p_x x_n(q') + p_y y_n(q')}{p_x x_n(q) + p_y y_n(q)}. 
\end{equation}
Therefore, the inequality   holds iff  $V_\textrm{withdraw} \ge V_\textrm{mint}$. %when the numerator is bigger than the denominator. 
% \begin{equation*}
 %   p_x x_n(q') + p_y y_n(q') \ge p_x x_n(q) + p_y y_n(q). 
 %\end{equation*}
 The result follows from   substituting Eq.~\eqref{eq:position-holdings} into the above expression, followed by  algebraic manipulation. %we compare both sides of the inequality. After simplifying and rearranging, the condition reduces to the stated inequality depending on whether $q < q'$ or $q > q'$  \dsm{is this last sentence necessary?}\bt{I dont know, probably not}.   The result then follows after algebraic manipulatin. 
\end{proof}
% \bt{Should we make proofs for the corollary as well ? They follow directly from the above theorem}

\addtocounter{equation}{1}

% \gf{I think it would be good to combine Thm 2 and Thm 3 into one statement with multiple parts---especially since the proof of Thm 3 is very short, it's more of a proposition if we leave it standalone. But it's important for the paper, so I'd like to have it be part of a Thm. I'm also having a hard time reconciling the two claims. Thm 2 says that if $q$ equals external price, then $\mathcal C\leq 0$. But in Thm 3, if I substitute in $q=p_x/p_y$, then I get that $\mathcal C\leq 0$ iff $\frac{p_x}{p_y} \geq \hat q'$. Which is true? Am I missing something?}

Based on this result, we can establish sufficient conditions for either favorable or adverse price impact. These are summarized in the following corollary.
% \gf{Let's combine these corollaries into one with two cases.} \dsm{Bruno: let us fix this}

\begin{corollary}[Gains under diverging prices  and losses under  converging prices]\label{coro2}
When initial and final prices do not cross $p_x/p_y$, gains and losses are determined by the direction of the price movement:
\begin{itemize}
    \item Gain under diverging prices: 
\( \mathcal{C} \le 0 \) if \( q' < q \leq  {p_x}/{p_y} \)  or \( q' > q \ge {p_x}/{p_y} \), i.e., if  a trade moves the AMM price \emph{away} from the external market price, the LP experiences gains.
\item Loss under converging prices:  \( \mathcal{C} \ge 0 \) if \( q < q' < {p_x}/{p_y} \)  or \( q > q' > {p_x}/{p_y} \),  i.e., if  a trade moves the AMM price \emph{towards} the external market price, the LP experiences losses.
\end{itemize}
\end{corollary}

\begin{comment}
    
\begin{corollary}[Loss under \gf{converging prices}]\label{coro1} 
\( \mathcal{C} \ge 0 \) if \( q < q' < {p_x}/{p_y} \)  or \( q > q' > {p_x}/{p_y} \), \gf{i.e., if the price of a trade moves the AMM price \emph{towards} the external market price}.
\end{corollary}
\end{comment}

\begin{remark}
    [No-loss if pool price initially aligned with market]\label{theorem1}
Note that if the pool price is initially aligned with the market price, i.e., \( q = {p_x}/{p_y} \), then any nonzero price movement necessarily moves the pool price away from the market-implied price. By the gain-under-diverging-prices clause of Corollary~\ref{coro2}, the price impact \( \mathcal{C} \) experienced by a liquidity provider is therefore non-positive: \( \mathcal{C} \le 0 \).
\end{remark}

% \begin{remark} 
% \dsm{According to Remark~\ref{theorem1}, if \( q = p_x / p_y \), then \( \mathcal{C} \leq 0 \). In Theorem~\ref{theorem2}, consider the case \( q' < p_x / p_y \). Then \( \mathcal{C} \leq 0 \) if and only if \( p_x / p_y \geq q' \), which holds by construction. Similarly, if \( p_x / p_y < q' \), then \( \mathcal{C} \leq 0 \) if and only if \( p_x / p_y \leq q' \), which again holds by construction. Therefore, Remark~\ref{theorem1} follows from Theorem~\ref{theorem2}.  }
% \gf{ok got it---then can we list thm 5 as a corollary of thm 6, so we swap the order? The proof of thm 6 currently depends on that of thm 5, but we could rearrange the proofs accordingly.}
% \end{remark}

% Taken together, these results clarify when price impact is negative (favorable to LPs) or positive (unfavorable). 
Intuitively, when the pool price moves toward the market price (the loss-under-converging-prices clause of Corollary~\ref{coro2}), traders obtain better execution, acquiring tokens from the pool at prices closer to fair market value—resulting in a loss for LPs. Conversely, when the pool price moves away from the market price (the gain-under-diverging-prices clause of Corollary~\ref{coro2}), traders face worse execution, effectively overpaying and enriching LPs. 
% This dynamic highlights the fundamental trade-off in liquidity provision and motivates the strategic timing and placement decisions of JIT LPs.

\subsection{The Effects of Liquidity}
We analyze the limiting behavior of price impact as  available liquidity changes. Suppose we have a price range $(a, b)$ which is currently covered by total liquidity $L$ by passive LPs. A JIT LP is considering adding liquidity to this range; intuitively, when liquidity \( L \) tends to zero, the position is too small to have any meaningful participation, and the price impact vanishes. Conversely, as \( L \to \infty \), the position becomes large enough that the price becomes almost static. Formally, we establish the following result
% whose proof is presented in Appendix~\ref{sec:prooflemalimitingbehavior}: \gf{broken link}

\begin{lemma}[Limiting behavior of price impact]\label{lemma:limit_price_impact}
    Let \( \Delta x \) be the quantity of token \( X \) exchanged in a trade. Consider \( q \in (t_m, t_{m+1}]\) the current pool price and consider \( ( L, a,b)\) to be the single position of some provider $n$ such that $q \in (a,b)$, in words, consider \( L \) to be all the liquidity minted before the transaction by some provider $n$. 
    % \gf{I thought we were going to let this be all the liquidity every minted by any LP?}\bt{But then it doesn't make that much sense, because impermanent loss is unique to each LP, right? What would mean for all the liquidity to go to 0 ? The trade wouldnt even occur in the first place.} 
    Then:
\[
\lim_{L \to \infty} \mathcal{C} = \Delta x \cdot (p_y \cdot q - p_x), \quad \text{and} \quad \lim_{L \to 0} \mathcal{C} = 0.
\]
\end{lemma}
% \gf{This statement is kind of confusing---the notation $\mathcal C$ seems like it should depend on $q$ and $q'$ as well as the prices $p_x, p_y$ and $L$. Also what does it mean for $L \to \infty$ in a CLMM?  Do you mean that for every interval, the liquidity in that interval tends to infinity? Or is this lemma considering a single interval, so $L$ is a scalar?}\bt{The thing here is that those lemmas, although generalizable, were initially thought to justify the JIT case. Therefore L here would be a scalar that represents the liquidity minted around the current price $q$. I made such things more explicit, please tell me if now its clear.}

% \begin{proof}[Proof sketch]
 %As \( L \to \infty \), the price becomes effectively fixed at \( q \), and the LP executes the trade at constant price \( q \), leading to \( \mathcal{C} \to \Delta x (p_y q - p_x) \). As \( L \to 0 \), the LP's position is infinitesimal, yielding negligible participation in the trade and thus \( \mathcal{C} \to 0 \). \end{proof}

 \begin{proof}[Proof sketch]
    (Full proof in Appendix \ref{sec:prooflemalimitingbehavior})
    The proof follows from Eq.~\eqref{eq:priceimpact1b} and the following facts derived in~\cite{UniswapV3Core} (see also Eq.~\eqref{eq:position-holdings}): $\frac{\Delta x}{L} \triangleq \Delta \frac{1}{\sqrt{q}} ={1}/{\sqrt{q'}}- {1}/{\sqrt{q}}$ and $\frac{\Delta y}{L} \triangleq \Delta \sqrt{q} = \sqrt{q'}- \sqrt{q}$. 
Then, \( \Delta y = -\Delta x \sqrt{q'} \sqrt{q} \). % we have:
     %\begin{equation}
    %\mathcal{C} = - (p_x \Delta x + p_y \Delta y).
    %\end{equation}
    As \( L \to \infty \), we have $q' \to q$.  %When \(q'  \approx q\),   within a  given  tick, we have: \( \Delta y \approx -\Delta x \sqrt{q'} \sqrt{q} \).  
    In the limit,    \( \Delta y \to - q \Delta x  \), and 
   $  \mathcal{C} \to -  p_x \Delta x + p_y  q \Delta x$ (see Eq.~\eqref{eq:priceimpact1b}). 
    As \( L \to 0 \), the LP's position is infinitesimal, yielding negligible participation in the trade and thus \( \mathcal{C} \to 0 \).
\end{proof}

% \begin{proof}[Proof sketch] From the definition of \( \mathcal{C} \), we observe that it scales linearly with \( L \), and that the difference \( | q - q'| \) decreases with increasing \( L \).   As \( L \to \infty \), the pool price becomes effectively constant during the swap, i.e., \( q' \to q \). In this limit, the LP's position executes the entire trade at price \( q \), resulting in \( \mathcal{C} \to \Delta x (p_y q - p_x) \).  \end{proof}

\begin{figure} \centering
\begin{tikzpicture}[scale=2]

    % Define colors
    \definecolor{lightblue}{rgb}{0.8, 0.93, 1}
    \definecolor{lightgreen}{rgb}{0.85, 1, 0.85}
    \definecolor{lightred}{rgb}{1, 0.85, 0.85}
    \definecolor{lightyellow}{rgb}{1, 1, 0.75}
    \definecolor{lightpurple}{rgb}{0.9, 0.85, 1}
    \definecolor{lightorange}{rgb}{1, 0.93, 0.8}

    % Regions (adapted to q* = 1 and 2x2 domain)
    \fill[lightred] (0,1) rectangle (1,2);            % Region I
    \fill[lightgreen] (1,1) rectangle (2,2);            % Region II
    \fill[lightred]   (1,1) -- (2,2) -- (2,1) -- cycle; % Red triangle in Region II
    \fill[lightred](1,0) rectangle (2,1);            % Region III
    \fill[lightgreen](0,0) rectangle (1,1);            % Region VI
    \fill[lightred](0,0) -- (0,1) -- (1,1) -- cycle; % Orange triangle in Region IV

    % Region VIII (q < 1, q' > 1, above q' = 1/q)
    \fill[lightgreen,domain=0.5:1, variable=\x]
        plot[smooth, samples=100] 
        (\x, {1/(\x)}) --
        (1,2) -- (0.5,2) -- (0.5,1) -- cycle;

    % Region IV (q > 1, q' < 1, below q' = 1/q)
    \fill[lightgreen,domain=1:2, variable=\x]
        (1,0.5) --
        plot[smooth, samples=100] 
        (\x, {1/(\x)}) --
        (2,1) -- (2,0) -- (1,0)  -- cycle;

    % Axes
    \draw[->] (-0.1, 0) -- (2.2, 0) node[right] {\( q \)};
    \draw[->] (0, -0.1) -- (0, 2.2) node[above] {\( q' \)};

    % Diagonal line q = q'
    \draw[dashed, thick] (0,0) -- (2,2) node[anchor=south west] {\( q = q' \)};

    % Vertical line at q = 1 (p_x / p_y)
    \draw[dotted, thick] (1,0) -- (1,2) node[anchor=south] {\(  q=p_x/p_y \)};

 \draw[dotted, thick] (0,1) -- (2,1) node[anchor=west] {\(  q'=p_x/p_y \)};
    
    \draw (1, 0) -- (1, -0.05) node[below] {\( 1 \)};

     \draw (0,1) -- ( -0.05,1) node[left] {\( 1 \)};

    % Curve q' = 1/q
    \draw[thick, black, domain=0.5:2, samples=100] 
        plot (\x, {1/(\x)}) node[pos=0.9, above left] { };

    % Region labels (adjusted positions)
    \node at (0.8, 1.7) {\textbf{I}};
    \node at (1.3, 1.7) {\textbf{II}};
    \node at (0.5, 0.8) {\textbf{VII}};
    \node at (1.6, 0.8) {\textbf{IV}};
    \node at (0.5, 1.2) {\textbf{VIII}};
    \node at (1.6, 1.2) {\textbf{III}};
       \node at (0.8, 0.5) {\textbf{VI}};
    \node at (1.3, 0.5) {\textbf{V}};

\end{tikzpicture}
\caption{Partition of the \((q, q')\)-space with \( {p_x}/{p_y} = 1 \); black curve corresponds to \( q' = {1}/{q} \). Red and green  regions correspond to \( \mathbf{PI} \ge 0 \) and \( \mathbf{PI} \leq 0 \), respectively. %\gf{Bruno, could you please explain in the caption what the figure means fully? Answer: see below}
}
\label{fig:partition}
\end{figure}
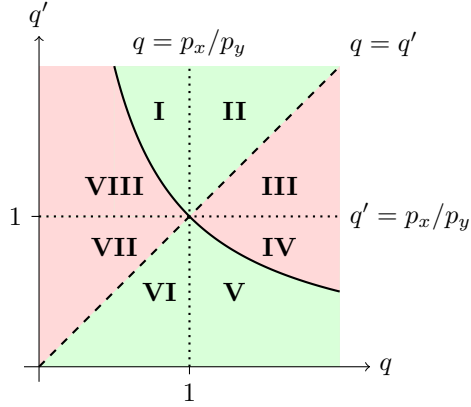
This result reinforces the intuition that liquidity acts as a dampener of price impact: as liquidity increases, the pool becomes less sensitive to trade-induced price shifts, and the value transfer is increasingly governed by the degree of misalignment between pool and market prices. To illustrate the broader implications of price dynamics and their relationship with price impact, we present the following example.

\begin{example}
Figure~\ref{fig:partition} provides a geometric interpretation of price impact in the \((q, q')\)-plane, where  the market price ratio is fixed at \( p_x/p_y = 1 \). The space is partitioned into colored regions: green for \( \mathcal{C} \le 0 \) (favorable to LPs), and red for \( \mathcal{C} \ge 0 \) (unfavorable to LPs). Corollary~\ref{coro2} provides sufficient conditions for LPs to experience non-positive price impact—Regions  VI and II.  Conversely, Corollary~\ref{coro2} also identifies conditions under which price impact is non-negative—Regions VII and III. %)—traders acquire tokens closer to market value, reducing LP gains.  
More intricate behavior arises when the price crosses the market value (Theorem~\ref{theorem:price_impact}).  For example, in Region~I (\( q' > p_x/p_y > q \) and \( q' > 1/q \)), price impact is favorable; in Region VIII (\( q' < 1/q \)), it becomes unfavorable. A symmetric situation arises for Regions IV and V depending on whether \( q' < 1/q \). 
%These results emphasize that the direction and extent of price movement—not just the endpoints—critically determine the outcome for liquidity providers.
\end{example}

Analogous to our analysis of price impact, we can show that the fee function is continuous with respect to the liquidity variable and converges to a finite limit as liquidity grows large. This result is formalized in the following lemma
% , whose proof is provided in Appendix~\ref{app:last_theorem}. \gf{broken link}

% Having established the fundamental properties of both price impact and fee functions, we now integrate these insights into a unified optimization framework. The following section formulates the decision problem faced by JIT LPs, capturing the trade-offs between earned fees and  incurred price impact.

% The following lemma establishes this connection and provides useful properties regarding the price change dynamics.

% \gf{How does this lemma apply to our story? Can you add a sentence or two to introduce it?}

\begin{lemma}[Effect of Liquidity]\label{lemma:continuity-monotonicity}
Let \(q\) be the current price in a CLMM, where $q \in (t_{m}, t_{m+1}]$. 
Suppose we have a target trade that pays $\Delta x$ $X$ tokens to the pool and asks for $Y$ tokens, which can be supported by the CLMM. 
Immediately before the trade, a JIT LP adds liquidity $L \in [0, \infty)$ to price range $(a, b)$ where $a \le t_m < t_{m+1} \le b$. 
Let $q' (< q)$ be the resulting price after the trade as a function of $L$. 
Then, $q'$ is continuous, strictly decreasing in $L$, and $\lim_{L \to \infty} q' = q$.
\end{lemma}
\begin{proof}[Proof sketch]
(Full proof in Appendix \ref{sec:proofcontinuitymonotonicity})
We model the final price $q'$ after the swap as a function of added liquidity $L$. The swap demand $\Delta x$ must be fulfilled across ticks, so we express $q'$ as an implicit function of $L$. As $L$ increases, the same $\Delta x$ has less price impact, so $q' \uparrow$. The mapping is strictly decreasing in $L$ and continuous due to the functional form of the AMM and the monotonicity of the inverse square root. As $L \to \infty$, $q' \to q$, meaning the AMM price becomes increasingly resistant to change. 
% The full demonstration can be found in \cite{extended}.
\end{proof}

\begin{remark}
    Note that, since $P_{m} = \sum_{i \in [N]} P_{n,m}$, Lemma~\ref{lemma:continuity-monotonicity}  is also true for any $P_{n,m}$; that is, $  \lim_{P_{n,m} \to \infty} q' = q \; \forall n$ such that \( t_m < q \le t_{m+1} \). This means that it  is sufficient for \textbf{any} provider to inject high amount of liquidity in order to force the limit condition.
    
\end{remark}
Lemma \ref{lemma:continuity-monotonicity} highlights one of the central messages of this work: by concentrating liquidity when and where it is needed, JIT LPs can dampen price volatility, controlling slippage and choosing how the price will change. This empowers the JIT LP to optimize their revenue by balancing earned fees and losses due to price changes (See \S\ref{sec:imperloss}).

\section{Just-In-Time Liquidity Allocation: \\An Optimization Perspective} \label{sec:jit-optimization}
% \gf{How about making this a full section? (Ie section 4)?}

% In this section, we use the results from Section \ref{sec:price-impact-conditions} to formulate the optimization problem faced by  JIT  liquidity providers. 

% Building on the insights from the preceding sections on price impact and fee distribution, 
We now define the formal structure of the decision process faced by a JIT LP observing a pending transaction. 
% The formulation captures the strategic behavior of a JIT LP aiming to maximize utility—defined as the net profit from fees earned, price impact incurred, and other  costs (e.g. bidding fees). 
We begin by explaining the relevant properties of JIT LPs (\S\ref{sec:back-jit}), which inform our model, including the strategy space available to the JIT LP (\S\ref{sec:stratsp}), and the utility function that governs decision-making. Finally, we  establish conditions under which an optimal strategy exists 
(\S\ref{sec:utilfun}).

\subsection{Background on Just-In-Time LPs}
\label{sec:back-jit}
Unlike  ordinary LPs who passively leave their liquidity in the liquidity pool, Just-In-Time (JIT) LPs react to market conditions in real-time to extract profit \cite{JIT}. We next present two important aspects of JIT LP behavior and explain how we model them.

\vspace{0.1in} \noindent \textbf{(1) Per-transaction optimization.~~} JIT LPs operate by monitoring the public \emph{mempool}\footnote{The mempool  is a public set of pending transactions propagated through the blockchain peer-to-peer network.} \cite{wang2024privateorderflowsbuilder, yukselretrospective} for trade opportunities.
This allows JIT LPs to selectively decide whether to provide liquidity for each trade before the trade is validated and finalized on the blockchain. 

Once a JIT LP identifies a favorable trade, it can submit a \textit{bundle}, which is a set of one or more transactions submitted together for inclusion in a block, typically via a relay to a block builder (e.g., via FlashBots \cite{Flashbots}). These bundles are guaranteed atomic execution: the entire bundle either succeeds or fails as a group, ensuring riskless, single-trade provisioning.
Exploiting this atomicity, the JIT LP can execute various operations-e.g., \emph{sandwich attacks}, which mint liquidity just before a swap and burn it immediately after \cite{cryptoeprint:2023/973}.  
% -while optimizing for short-term profit \cite{cryptoeprint:2023/973}.  

\textit{Modeling Implications:}
JIT LPs exploit this information asymmetry to compute their expected utility for a given trade.
% with high precision. 
Based on this foresight, they can make optimal, \emph{transaction-level} decisions about whether to mint and burn liquidity, factoring in both trading fees and the impact of price movements.

This naturally induces a two-stage optimization framework: first, passive LPs select their liquidity allocation strategy for a longer period \cite{tang2024game}; second, for each observed trade, JIT LPs react by solving an online optimization problem to maximize net expected revenue, potentially adjusting for inclusion costs such as gas fees or auction bids.
Our model focuses on this second optimization.

\vspace{0.1in} \noindent \textbf{(2) Competition among JIT LPs.~~}
With multiple JIT LPs potentially competing over the same trade, each constructs and submits a bundle. 
% \dsm{which is a set of one or more transactions submitted together for inclusion in a block, typically via a relay to a block builder. }  
The bundle includes a tip to the block builder, 
% \gf{Explain what a bundle is} \dsm{done} \dsm{The builder 
who selects bundles to maximize the total value extractable from the block by adjusting transaction ordering, inserting additional transactions, or censoring transactions. 
 % \gf{Have we defined MEV?}  \dsm{ removed the reference to MEV} 
% This creates a sealed-bid auction: participants bid without visibility into others’ bids, and only the highest-bidding bundle is included.

This creates a sealed-bid auction: each JIT LP submits a bundle with a bid, without visibility into others’ bids, and the builder selects the one that maximizes their expected revenue, i.e., only the highest-bidding bundle is included. The winning JIT LP earns a portion of the fees and pays a tip bounded below by the second-highest bid. This mechanism forms the basis of our transaction-level utility model, where we analyze how JIT LPs optimize their strategy under fixed surplus and uncertain inclusion.

\textit{Modeling Implications:}  Modeling this, recall that for a swap of $\Delta x$ tokens of $X$ and fee rate $\alpha$, 
 % let \( \Delta x \) denote the amount of token \( X \)  exchanged in the observed trade. Given a fee rate \( \alpha \), 
 the total fee generated by the swap is \( \delta = \alpha \cdot p_x \cdot \Delta x \), where \( p_x \) is the market price of token \( X \) in dollars. Since both \( \Delta x \) and \( p_x \) are visible in the mempool, \( \delta \) is known to all JIT LPs in advance. As a result, they compete over a fixed surplus, strategically offering portions of \( \delta \) as tips to the block builder to secure inclusion.

\subsection{Model}
Formally, we represent the JIT optimization problem as a tuple \( ([N], \mathcal{U}, \mathcal{S}, \theta, \rho) \).
% where each component defines a key element of the decision framework. 
Here, \([N]\) denotes the set of all passive liquidity providers in the pool. The strategy space \(\mathcal{S}\) encapsulates all feasible actions available to the JIT LP, while \(\mathcal{U}\) is the utility function mapping each strategy to a real-valued profit—computed as the difference between earned fees, incurred price impact, and bidding costs. The vector \(\theta\) includes all relevant swap parameters, such as trade size, fee rate, price information, and the liquidity distribution of passive providers. Finally, $\rho$ is a constant that  denotes the JIT budget.

% Finally, \(\rho\) denotes the JIT LP's budget constraint, ensuring that any admissible JIT strategy respects the LP’s capital limitation (\S).

Note that the JIT optimization problem assumes that the passive LP positions are established a priori. We denote by  \(\boldsymbol{s}\)   the strategy profile of passive LPs, encoded as a vector of liquidity allocations across ticks: \(\boldsymbol{s} = \left(L_i, a_i, b_i\right)_{i \in [N] }\). 
The aggregate effect of the passive LP strategies \(\boldsymbol{s}\) induces a liquidity distribution function \(P = (P_m)_{m \in M}\), where each \(P_m\) denotes the total passive liquidity available at tick \(m\). Formally, \(P\) is derived via a deterministic mapping from $\boldsymbol{s}$.
% The aggregate effect of these strategies defines a liquidity distribution function \(P = (P_m)_{m \in M}\), which describes the total passive liquidity at each tick \(m\). 
In the remainder of this section, we use \(\boldsymbol{s}\) to refer to the underlying strategies and \(P\) to denote the induced state of the pool.  %, noting that for our purposes the latter suffices as  input to the JIT optimization problem.

% In what follows, we treat \(P\) as the sufficient input summarizing the state of the pool from the perspective of the JIT LP.

\subsubsection{Strategy Space \(\mathcal{S}\)} \label{sec:stratsp}

A JIT LP’s choice of action consists in selecting one position\footnote{We do not consider multiple-position actions for two main reasons: 1) it is commonly assumed in the literature for JIT to mint constant liquidity \cite{cryptoeprint:2023/973}, and 2) the data provided in \cite{tang2024game} reveals that  transactions from 2024 to 2025 by JIT LPs include only a single position.} \((L, a, b)\) to be minted immediately before a swap and burned immediately after.  Since the entire state of the pool is observable at decision time, the JIT LP can make an informed choice based on current market and pool conditions. To formalize this, we define the swap as a 6-tuple \(\theta \triangleq (\Delta x, q, P, p_x, p_y, \alpha)\), where:
\begin{itemize}
    \item \(\Delta x\) is the amount of token \(X\) being exchanged by the trader in the target trade;
    \item \(q\) is the current pool price;
    \item \(P = (P_m)_{m \in M}\) is the per-tick liquidity distribution from passive LPs;
    \item \(p_x\) and \(p_y\) are the external market prices (in dollars) of tokens \(X\) and \(Y\), respectively;
    \item \(\alpha\) is the pool's fee rate.
\end{itemize}

Let $q^*$ denote the price after the swap without the presence of JIT LPs. 
We define   \newline
\begin{comment}
    \noindent\scalebox{0.85}{\quad
\[ 
% \mathcal{R}^{(q, q^*)} \triangleq \mathcal{R} \cap \Big\{ (a, b) \Big| [a, b] \subseteq \left[\max\big(T \cap (-\infty, \min\{q, q^*\}]\big), \min\big(T \cap [\max\{q, q^*\}, \infty) \big) \right] \Big\}. 
\mathcal{R}^{(q, q^*)} \triangleq \mathcal{R} \cap \Big\{ (a, b) \Big| \forall m: a \le t_m < t_{m+1} \le b, ~ [t_m, t_{m+1}] \cap (\min\{q, q^*\}, \max\{q, q^*\}) \neq \varnothing \Big\}.
\]
}
\end{comment}
\noindent
\scalebox{0.85}{%
  \parbox{\linewidth}{%
 \quad  \[ \quad
    \mathcal{R}^{(q, q^*)} \triangleq 
    \mathcal{R} \cap \Big\{ (a, b) \Big| 
    \forall m: a \le t_m < t_{m+1} \le b, ~ 
    [t_m, t_{m+1}] \cap (\min\{q, q^*\}, \max\{q, q^*\}) \neq \varnothing 
    \Big\}.
  \]
  }%
}

Intuitively, $\mathcal{R}^{(q, q^*)}$ includes all feasible price ranges that do not contain a subrange $[t_m, t_{m+1}]$ that is never ``touched'' when price moves from $q$ to $q^*$. 
We claim that the JIT LP only considers price ranges in $\mathcal{R}^{(q, q^*)}$ for their choice of liquidity position. Otherwise, 
%by Remark \ref{rmk:associative-law}, 
we can break the position $(L, a, b)$ down to two positions $(L, a, c)$ and $(L, c, b)$ where one of their ranges does not intersect with the price moving range $(\min\{q, q^*\}, \max\{q, q^*\})$. Let it be $(a, c)$ without loss of generality \cite{tang2024game}. 
During the trade, $(L, a, c)$ will incur fee income $\mathcal{F} = 0$ and zero token amount change, which further implies price impact $\mathcal{C} = 0$ since $p_x$ and $p_y$ are both constant. Hence, in comparison with an action that chooses $(L, c, b)$ instead, the $(L, a, b)$ has equal utility impact while wasting budget on $(a, c)$. 
 Accordingly, we define the JIT strategy space as:
\[
\mathcal{S} = \left\{ (L, a, b) \;\middle|\;
(a,b) \in \mathcal{R}^{(q, q^*)},\;
L \cdot V_{(a,b)} \le \rho,\;
q \in [a,b]
\right\} \cup \{ \perp \},
\]
where \( \perp \) denotes non-participation, $\rho$ denotes the LP's budget, and \( V_{(a,b)} \) is the dollar cost per unit of liquidity at price $q$ over the interval \( (a,b) \).

\subsubsection{Utility Function} \label{sec:utilfun}

The objective of the JIT LP is to maximize a utility function \( \mathcal{U} \), defined as the net profit obtained from participating in a given swap. This profit captures the trade-off between fees earned and the price impact incurred. Given a strategy \( s = (L, a, b) \in \mathcal{S} \) and swap parameters \( \theta \), 
% \gf{point to sec 2 to remind the definition)}\bt{\(\theta\) is not defined in section, is defined in this same page on the top)}, 
the utility is defined as:
\begin{equation}
    \mathcal{U}(s; \theta) = \mathcal{F}(s; \theta) - \mathcal{C}(s; \theta),
\end{equation}
where \( \mathcal{F}(s; \theta) \) is the total fee income and \( \mathcal{C}(s; \theta) \) denotes the price impact associated with the strategy. The JIT LP aims to solve:
\begin{equation}
    s^* = \arg \max_{s \in \mathcal{S}} ~\mathcal{U}(s; \theta).
\end{equation}

% \begin{remark}
%\wt{Could you review this, Rachid?}
In practice, a JIT LP needs to additionally pay a cost $v$ for placing bids in auctions (such as Flashbots~\cite{Flashbots}) to land the sandwich attack on the trade transaction. 
Ideally, we should incorporate the auction in our model. Two natural approaches are (1) to model the game among multiple JIT LPs, or (2) model the probability of winning the bid as a function of $v$. 
Both approaches require empirical study on historical auctions that involves not only auction winners, but more importantly, losers. 
Such data is not currently accessible, to the best of our knowledge, and it is unrealistic to make ad hoc assumptions.
%For example, the game being fully competitive such that $v$ equals the max utility, or an ad hoc prior distribution of bid winning. \gf{The previous sentence is confusing and not grammatical---can we rephrase?}
Hence, for simplicity, we assume $v$ to be a constant that does not relate to the choice of liquidity position. As a result, we can easily extend our model by updating the budget constraint as  \( L \cdot V_{(a,b)} + v \le \rho \) and the utility function as \( \mathcal{U}(s; \theta) = \mathcal{F}(s; \theta) - \mathcal{C}(s; \theta) - v \). 
% As noted in Remark~\ref{remark1}, bidding costs—such as gas fees or tips to block builders (see —can be incorporated by extending each strategy to a 4-tuple \( (L, a, b, v) \), where \( v \) is the bid amount. The utility function is updated to , and the budget constraint becomes \( L \cdot V_{(a,b)} + v \le \rho \). Because flashbots auctions are not trivial to understand, since they consider other participants than simply other JITs, we defer their analysis and consider a fixed price at each transaction. Thus, \( v \) is treated as constant, and all results in this section remain valid under this generalization.
% \end{remark}

Given a specific strategy \( s \in \mathcal{S} \) and the current swap context, we can compute the total utility associated with executing the strategy. From this point forward, we omit explicit dependence on \( \theta \) and \( s \) whenever the context is clear.

The total fee earned by the JIT LP is decomposed tick-wise. Since the JIT LP deploys a single position \( (L, a, b) \), let $L$ be the liquidity of the JIT on tick $m$.  The JIT’s share of the fees at tick \( m \in M \) is:
\begin{equation}
\mathcal{F}_m \triangleq \underbrace{\Delta x_m \alpha p_x}_{\text{Fee }=\delta_m} \cdot \frac{L}{\sum_{i \in [N]} P_{i,m} + L}, \quad \quad
\mathcal{F} \triangleq \sum_{m \in M} \mathcal{F}_m.
\end{equation}

 Let \( t_m \) denote the price at tick \( m \). Recall that \(\alpha, p_x\) are trade parameters, and that \( \Delta x_m\) depends on \(q'\) and denotes the amount of token \(X\)   that is injected at each tick \(m\) (see~\S\ref{sec:fees}). As  the swap affects only the ticks crossed during execution, the tick-level fee \( \delta_m \) is nonzero only for ticks \( m \) such that \( t_m \in (q, q') \). Under these conditions, the fee allocation simplifies to:
\begin{equation} \label{eq:feeTick}
\mathcal{F}_m = \delta_m \cdot \frac{L}{\sum_{i \in [N]} P_{i,m} + L}, \quad \quad
\mathcal{F} = \sum_{m \in M \,\mid\, t_m \in (a,b)} \mathcal{F}_m.
\end{equation}
From it, we can show the following lemma:

\begin{lemma}\label{lemma:limit-fee}
Consider an LP $n$ with a single position $(L,a,b)$, such that \(q \in (a, b)\). Let there be a swap that pays \( \Delta x > 0 \) $X$ tokens. The total fee earned by the provider is given by:
\[
\mathcal{F}(L) = \sum_{\forall m| t_m \in (a,b)} \delta_m \cdot \frac{L}{L + {\tilde{P}_{n,m}}},
\]
where \( \tilde{P}_{n,m} = \sum_{i \not= n } P_{i,m}\) is the total liquidity in tick $m$ without $ L = P_{n,m} $. In addition, \( \mathcal{F}(L) \) is continuous in \( L \) and is bounded above by \( \alpha  \cdot \Delta x \),
\begin{equation}
    \lim_{L \rightarrow \infty} \mathcal{F}(L) \leq \alpha \cdot \Delta x. 
\end{equation}
\end{lemma}
\begin{proof}[Proof sketch]
(Full proof in Appendix \ref{sec:limitfeeproof})
We express the total fee earned by the LP as a sum over ticks, weighted by the LP’s share of liquidity \eqref{eq:feeTick}. Since each fee component depends continuously on $L$ through the liquidity share, the full fee function is continuous in $L$. As $L \to \infty$, the LP's share of each tick approaches 1, but the total fee is still bounded above by $\alpha \cdot \Delta x$, since that is the maximum total fee generated by the swap. Hence, the limit exists and is bounded.
\end{proof}

The price impact \( \mathcal{C} \) is likewise decomposed across ticks as:
\begin{equation}
\mathcal{C}_m \triangleq p_x \cdot \left( x(\hat{q}_m) - x(\hat{q}_m') \right) + p_y \cdot \left( y(\hat{q}_m) - y(\hat{q}_m') \right), \quad 
\mathcal{C} \triangleq \sum_{m \in M \,\mid\, t_m \in (a,b)} \mathcal{C}_m,
\end{equation}
where \( \hat{q}_m = \min\{t_{m+1}, \max\{q, t_m\}\} \) is the entry price, \( \hat{q}_m' \) the exit price within tick \( m \) and \(x(q)\),\;\(y(q)\) are given by Eq.~\eqref{eq:position-holdings}.

This decomposition allows the utility to be expressed in tick-wise form:
\begin{equation}\label{tickUtility}
    \mathcal{U}_m = \mathcal{F}_m - \mathcal{C}_m, \quad 
    \mathcal{U} = \sum_{m \in M \,\mid\, t_m \in (a,b)} \mathcal{U}_m.
\end{equation}

Recall that all parameters \(\theta \triangleq (\Delta x, q, P, p_x, p_y, \alpha)\) are fixed by the given trade, and since \(q'\) depends on \((L,a,b)\), the only free variables are \(L,(a,b)\).
Using Lemmas~\ref{lemma:continuity-monotonicity}, \ref{lemma:limit_price_impact}, and \ref{lemma:limit-fee}, we establish the following theorem:
% (Proof in \S\ref{app:last_theorem}): \gf{2 broken links}

\begin{theorem}\label{theorem:utility_opt}
    The utility function \( \mathcal{U} \), as defined in~Eq.~\eqref{tickUtility}, attains a global maximum over the strategy space \( \mathcal{S} \). That is, there exists an optimal strategy \( s^* \in \mathcal{S} \) such that \( \mathcal{U}(s^*) \ge \mathcal{U}(s) \) for all \( s \in \mathcal{S} \).
\end{theorem}

\begin{proof}[Proof sketch]
(Full proof in Appendix \ref{sec:prooflasttheo}) The utility function $\mathcal{U}$ depends on liquidity $L$ and the chosen price range $(a,b)$. Since tick prices are discrete, there are finitely many such intervals. For a fixed interval, $\mathcal{U}(L)$ is continuous due to the continuity of both fee ($\mathcal{F}$) and cost ($\mathcal{C}$) functions in $L$. Also, $\mathcal{F}$ is bounded above while $\mathcal{C}$ grows at most linearly in $L$, constrained by a budget $\rho$. Therefore, the domain of feasible $L$ is compact, and $\mathcal{U}$ attains a maximum on it. Optimizing over all intervals proves existence of an optimal strategy.
 
\end{proof}
This guarantees that for any swap, there exists a profit-maximizing strategy for the JIT LP. In the following section, we propose an algorithm to compute this strategy and empirically explore how it behaves across different swap conditions.
% \gf{Is it unique?}\bt{Problably is but is not proved.}\gf{We said in the abstract we prove uniqueness, shall we remove that claim?}\bt{I removed, maybe we extend the proof for future versions in the arxiV}

\section{An Algorithmic Solution to the JIT \\ Optimization Problem} \label{sec:algo}

We now present an algorithmic approach for computing the optimal strategy \( s^* = (L^*, a^*, b^*) \) that maximizes the utility \( \mathcal{U} \) of a JIT LP in response to a given target swap. This is made possible by Theorem~\ref{theorem:utility_opt}, which guarantees the existence of a global optimum.

The utility function can be challenging to optimize since there are many parameters and the utility function is nonconcave in general.
% combinations that can significantly change the behavior of the function. 
In addition, we lack  a closed-form formula for \(q'\), which is calculated algorithmically. 
To address these challenges, we use empirical observations to reduce the search space.
First, we note that $\theta$, the fixed trade parameters, are known beforehand by the JIT trader; this turns the utility into a function of three variables: \(a,b,L\).

% While the search space is large,  our empirical data shows that most trades cross relatively few ticks. 
% % making the algorithm tractable in practice. 
% We further take advantage of the fact that  $\theta$, the fixed trade parameters, are known beforehand by the JIT trader; this turns the utility into a function of three variables: \(a,b,L\). 
% This is a lot easier to handle, but still not clear 
% The main remaining challenge is how to  optimize over \(a,b\). 
% \gf{Equations (14)-(16) as written do not appear to depend on $\theta, a,b,L$, so the previous 2 sentences are very confusing. }\bt{How so? They are defined as the sum accross (a,b), they directly mention L and other parameters in \(\theta\)}

% Here we take advantage of the practical case once again, in theory it could be the case that there are too many \((a,b)\) in the search space, this would make the problem hard. However, in real case scenarios,  it is the case that most transactions do not cross that many ticks, which makes it feasible for us to only try all the possible \((a,b)\).

In theory, the JIT trader could choose any $(a,b)$ interval, of which there are $O(M^2)$. 
However, due to the monotonicity of $q'$ in the liquidity, 
% in our data (Section \ref{sec:data-analysis})  most transactions do not cross many ticks; in particular, 
the original price $q$ cannot move farther than \( q^* \), where $q^*$ is the anticipated post-swap price in the \emph{absence} of JIT intervention. 
% which makes it feasible for us to only try all the possible \((a,b)\).
Hence, we limit ourselves to tick ranges \( (a, b) \in \mathcal{R}^{(q, q^*)} \). 
% where \( q \) is the current pool price. 
While this does not reduce the asymptotic worst-case complexity of our search, in our data  we observe that for 95\% of transactions, there are at most $2$ ticks between $q$ and $q^*$ (see \S\ref{sec:data-analysis}). Hence, enumerating the intervals in \(\mathcal{R}^{(q, q^*)} \) is practically feasible.

\begin{algorithm}[t]
\begingroup
%\footnotesize
\caption{JIT Optimal Strategy Search}
\label{alg:jit-strategy}
\begin{algorithmic}[1]
\Require Swap parameters \( \theta = (\Delta x, q, P, p_x, p_y, \alpha) \)\; Budget \( \rho \)
\Ensure Optimal strategy \( s^* = (L^*, a^*, b^*) \)

\State Initialize \( \mathcal{U}_{\max} \gets -\infty \) \label{line:init-utility}
\State Set \( s^* \gets \perp \) \label{line:init-strategy}
\For{each \( (a, b) \in \mathcal{R}^{(q, q^*)} \)} \label{line:loop-ranges}
    \State Solve for optimal liquidity: \label{line:solve-L}
    $    L^*(a,b;\theta)$ [Optimization \eqref{eq:MaxValue}]
    \State Evaluate resulting utility: \label{line:evaluate-U}
    $
    \mathcal{U} \gets \mathcal{F}^{(a,b)}(L^*;\theta) - \mathcal{C}^{(a,b)}(L^*;\theta)
    $
    \If{ \( \mathcal{U} > \mathcal{U}_{\max} \) } \label{line:check-best}
        \State Update best utility: \( \mathcal{U}_{\max} \gets \mathcal{U} \) \label{line:update-best-U}
        \State Update best strategy: \( s^* \gets (L^*, a, b) \) \label{line:update-best-s}
    \EndIf
\EndFor
\State \Return \( s^* \) \label{line:return}
\end{algorithmic}
\endgroup
\end{algorithm}

Algorithm~\ref{alg:jit-strategy} provides a computational method to search for the optimal JIT  liquidity provisioning strategy in response to an observed swap. 
The algorithm iterates over all candidate price intervals \( (a, b) \in \mathcal{R}^{(q, q^*)} \) (line~\ref{line:loop-ranges}).
For each interval $(a,b)$, it solves a one-dimensional, non-linear, non-concave 
 % \gf{is this maximizing a concave objective? Let's briefly mention how easy/hard it is to solve this optimization.} \bt{INcluded some disclaimers at the end, what you think of it ?}
 optimization problem to identify the liquidity level \( L^* \) that maximizes the JIT LP’s utility:
 
\begin{align}
\begin{aligned}
\label{eq:MaxValue}
L^*(a,b;\theta) = \arg\max_{L} \quad &\left(\mathcal{F}^{(a,b)}(L;\theta) - \mathcal{C}^{(a,b)}(L;\theta)\right) \\
  \textrm{s.t.} \quad & \underbrace{p_x x^{(a,b)}(q') + p_y y^{(a,b)}(q') \le \rho}_{\text{Budget Constraint}}
\end{aligned}
\end{align}
% subject to the budget constraint \(p_x x^{(a,b)}(q') + p_y y^{(a,b)}(q') \le \rho\)
where the utility and cost are defined as
\small
\begin{align}
\mathcal{F}^{(a,b)}(L;\theta) &= \sum_{m \in M \mid t_m \in (a,b)} \Delta x_m \cdot \alpha \cdot \frac{L}{L + P_m} \cdot p_x \\
\mathcal{C}^{(a,b)}(L;\theta) &= \sum_{m \in M \mid t_m \in (a,b)} \left( -p_x \Delta x_m + p_y \Delta x_m \hat{q}_m \cdot \frac{L+P_m}{\Delta x_m \sqrt{\hat{q}_m} + L + P_m} \right) \cdot \frac{L}{L+P_m}.
\end{align}
\normalsize

If the resulting utility exceeds the current maximum, the optimal strategy is updated accordingly (lines~\ref{line:check-best}–\ref{line:update-best-s}). After examining all admissible ranges, the algorithm returns the best strategy \( s^* = (L^*, a^*, b^*) \) (line~\ref{line:return}) that maximizes the JIT LP’s profit in the given swap scenario. 
 The optimal liquidity pair \( (a^*, b^*) \) and the corresponding optimal liquidity \( L^* \) are given by
$L^* = \arg \max_L   \mathcal{F}^{(a^*,b^*)}(L;\theta) - \mathcal{C}^{(a^*,b^*)}(L;\theta)$ such that \( \mathcal{U}^{(a^*,b^*)}(L^*;\theta) = \mathcal{U}(s^*;\theta) \ge \mathcal{U}(s;\theta) \; \forall \; s\in \mathcal{S}\), where  \( s^* \) is given by  the triplet $s^* = (L^*, a^*, b^*)$ (see \S\ref{sec:stratsp}).

  As long as we can  find the global optimum of a bounded continuous univariate function, 
  % \gf{what does "specific search space" mean? Compact and closed? Can we be more precise?}\bt{By specific search space I only meant the set of feasible strategies given the parameter \(\theta\) and the budget, I dont think its good phrasing, check if now is better} 
  Algorithm~\ref{alg:jit-strategy} returns an optimal strategy to maximize JIT utility.
  In our experiments, we tried both Particle Swarm Optimization (PSO) and binary search to solve for $L^*$. Since in most cases the optimal solution was to use the entire budget, both algorithms found the best strategy, but in non-trivial cases, PSO performed better at the cost of more computational power. 
  Our implementation of Algorithm \ref{alg:jit-strategy} can be found on \hyperlink{https://github.com/brunoCCOS/JITUniswapOptimization}{GitHub}.\footnote{https://github.com/brunoCCOS/JITUniswapOptimization} 
  % \bt{I made a hyperlink and a footnote for the cases where hyperlinks are not clickable}
  % \bt{I dont have the numbers comparing the two, dont know if it is in the scope of the paper. Please tell me maybe we can ommit this sentence}
  % \gf{I think it's fine for now. We can expand in journal version.}
  
%  The algorithm iterates over all candidate price intervals \( (a, b) \in \mathcal{R}^{(q, q^*)} \) (line~\ref{line:loop-ranges}), where \( q \) is the current pool price and \( q^* \) is the anticipated post-swap price in the absence of JIT intervention. For each candidate interval, it solves a one-dimensional, non-linear, non-concave 
%  % \gf{is this maximizing a concave objective? Let's briefly mention how easy/hard it is to solve this optimization.} \bt{INcluded some disclaimers at the end, what you think of it ?}
%  optimization problem to identify the liquidity level \( L^* \) that maximizes the JIT LP’s utility.

%  \begin{equation}
% \label{eq:MaxValue}
% L^*(a,b;\theta) = \arg\max_{L} \left(\mathcal{F}^{(a,b)}(L;\theta) - \mathcal{C}^{(a,b)}(L;\theta)\right)
% \end{equation}
% subject to the budget constraint \(p_x x^{(a,b)}(q') + p_y y^{(a,b)}(q') \le \rho\)

% \begin{align}
% \mathcal{F}^{(a,b)}(L;\theta) &= \sum_{m \in M \mid t_m \in (a,b)} \Delta x_m \cdot \alpha \cdot \frac{L}{L + P_m} \cdot p_x \\
% \mathcal{C}^{(a,b)}(L;\theta) &= \sum_{m \in M \mid t_m \in (a,b)} \left( -p_x \Delta x_m + p_y \Delta x_m \hat{q}_m \cdot \frac{L+P_m}{\Delta x_m \sqrt{\hat{q}_m} + L + P_m} \right) \cdot \frac{L}{L+P_m}.
% \end{align}

\begin{remark}[Token Directionality]\label{rmk:direction}
All derivations in this section assume \( \Delta x > 0 \), i.e., token \( X \) is being sold. To handle the symmetric case \( \Delta y > 0 \), where token \( Y \) is sold instead, we define a transformation $
Y(\theta) = (\Delta y, 1/q, P, p_y, p_x, \alpha)$. This mapping ensures that all formulas and results continue to hold under the appropriate relabeling of tokens.
\end{remark}

% Therefore, the algorithm is fully general and applicable regardless of the direction of the trade. It can be used to compute optimal strategies in both \( \Delta x > 0 \) and \( \Delta x < 0 \) scenarios, simply by applying the appropriate transformation to the swap parameters. %, as discussed earlier.

\subsection{Strategic Archetypes in CLMM Swaps}
\label{sec:types-of-trades}

% To make sense of the wide variety of possible swap configurations, 
We classify trades into three canonical scenarios—referred to as \textit{strategic \\ archetypes}—that capture the dominant patterns and analytical results of Theorems~\ref{theorem:price_impact} and~\ref{theorem:utility_opt}. See Appendix~\ref{app:images} for actual examples of each archetype.

\begin{itemize}
    \item \textbf{Overpriced trade:} The trade moves the pool price away from the market price, resulting in the purchase of the more expensive token. Formally, this occurs when $q' < q \le p_x/p_y$ or $q' > q \ge p_x/p_y$. According to Corollary~\ref{coro2}, such movements are associated with non-positive price impact, meaning the LP benefits (\Cref{fig:vertical_combined} (a1)). % These trades are favorable to JIT LPs since they can sell overpriced tokens and extract value from the misalignment.

    \item \textbf{Arbitrageur trade:} The trade brings the pool price closer to the market price by acquiring the cheaper token. This corresponds to $q < q' \le p_x/p_y$ or $q > q' \ge p_x/p_y$. As shown in Corollary~\ref{coro2}, this condition implies positive price impact, i.e., a loss for LPs. Although this trade is optimal for arbitrageurs, it is typically unprofitable for JIT LPs unless the fees compensate for the loss (\Cref{fig:vertical_combined} (a2)). 

    \item \textbf{Overshoot trade:} A corrective trade  crosses the market price, starting by buying the cheaper token but pushing the pool price past the market value, ultimately overpaying. This configuration—$q < p_x/p_y < q'$ or $q > p_x/p_y > q'$—lies on the boundary conditions specified in Theorem~\ref{theorem:price_impact}. In this case, whether the price impact is favorable or not depends on whether the post-trade price $q'$ crosses the threshold condition in \eqref{ineq1} or \eqref{ineq2} (\Cref{fig:vertical_combined} (a3)). % Overshoot trades are nuanced: JIT LPs can often engage only in the favorable segment of the trade to avoid the loss-inducing region.
\end{itemize}
% \gf{Can we have a simple figure illustrating these 3? It would make it easier to understand.}\bt{Check figure 5 please}

% To aid intuition, 

Recall that $q$ is the amount of token $Y$ per unit of token $X$, i.e., the pool price of token $X$ denominated in token $Y$; equivalently, $1/q$ is the pool price of token $Y$ denominated in token $X$. When $q > p_x/p_y$, token $X$ is overpriced in the pool relative to the market, implying that token $Y$ is relatively cheaper on-chain.

% Recall that $q$ is the pool price of token $Y$ in terms of token $X$, so $1/q$ is the pool price of $X$ in $Y$. When $q > p_x/p_y$, token $Y$ is overpriced in the pool relative to the market, implying that $X$ is relatively cheaper on-chain.

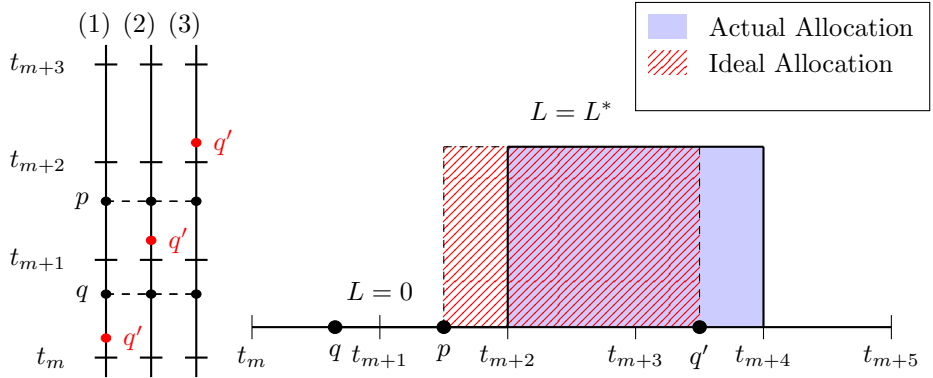
\begin{figure}[t]
\noindent
\begin{minipage}[t]{0.20\textwidth}
\centering
\begin{tikzpicture}[xscale=1.5, yscale=1.3]

% Horizontal positions for the 3 scenarios
\def\xone{-0.4}
\def\xtwo{0.0}
\def\xthree{0.4}

% Define common parameters
\def\q{0.65}
\node[left=3pt] at (\xone,\q) {\(q\)};
\def\p{1.6}
\node[left=3pt] at (\xone,\p) {\(p\)};

% Define q' positions (vertical now)
\def\qprimeOne{2.2}   % above p
\def\qprimeTwo{1.2}   % between q and p
\def\qprimeThree{0.2} % below q

% Draw tick marks and labels only once, along the dashed line
\foreach \y/\lbl in {0/t_m,1/t_{m+1},2/t_{m+2},3/t_{m+3}} {
    \draw[thick] (\xone-0.1,\y) -- (\xone+0.1,\y);
    \draw[thick] (\xtwo-0.1,\y) -- (\xtwo+0.1,\y);
    \draw[thick] (\xthree-0.1,\y) -- (\xthree+0.1,\y);
    \node[left=3pt] at (\xone-0.2,\y) {\(\lbl\)};
}

% Function to draw one vertical line with unlabeled q, p, labeled q'
\newcommand{\drawvertline}[4]{
    \def\x{#1}
    \def\qprime{#2}
    \def\colorqprime{#3}
    \def\number{#4}

    \draw[thick] (\x,-.2) -- (\x,3.2);

    % Draw points q and p as black dots, no labels
    \filldraw[black] (\x,\q) circle (1.2pt);
    \draw[dashed] (\xone,\q) -- (\xthree,\q);
    
    \filldraw[black] (\x,\p) circle (1.2pt);
    \draw[dashed] (\xone,\p) -- (\xthree,\p);

    % Draw q' point with label
    \filldraw[\colorqprime] (\x,\qprime) circle (1.2pt) node[right=3pt] {$q'$};
    
    \node[left=3pt] at (\x+0.2,3.4) {\(\number\)};
}

\drawvertline{\xone}{\qprimeThree}{red}{(1)}  % q' < q
\drawvertline{\xtwo}{\qprimeTwo}{red}{(2)}    % q < q' < p
\drawvertline{\xthree}{\qprimeOne}{red}{(3)} % q' > p

\end{tikzpicture}
\end{minipage}
\hspace{0.03\textwidth}
% Allocation diagram
\begin{minipage}[t]{0.20\textwidth}
\centering
\begin{tikzpicture}[xscale=1.7, yscale=1.6]

\def\L{1.5}
\def\q{0.65}

% Changed the loop to end at 5 (t_{m+5})
\foreach \x/\lbl in {0/t_m,1/t_{m+1},2/t_{m+2},3/t_{m+3},4/t_{m+4},5/t_{m+5}} {
    \draw (\x,0.1) -- (\x,-0.1);
    \node at (\x,-0.25) {\(\lbl\)};
}

% Adjusted points accordingly
\coordinate (q) at (\q,0);
\coordinate (p) at (1.5,0);
\coordinate (qp) at (2,\L);
\coordinate (qprime) at (3.5,0);       % Moved from 4.5 to 4 to fit new range
\coordinate (qprimeL) at (3,\L);
\coordinate (endActual) at (4,\L);

\fill[blue!20] (2,0) -- (2,\L) -- (4,\L) -- (4,0) -- cycle;
\fill[red!20, pattern=north east lines, pattern color=red] 
    (1.5,0) -- (1.5,\L) -- (3.5,\L) -- (3.5,0) -- cycle;

\draw[thick] (\q,0) -- (2,0);
\draw[dashed] (1.5,0) -- (1.5,\L);
\draw[dashed] (3.5,\L) -- (3.5,0);
\draw[dashed] (1.5,\L) -- (3,\L);
\draw[thick] (0,0) -- (5,0); 

\draw[thick] (2,0) -- (2,\L);
\draw[thick] (2,\L) -- (4,\L);
\draw[thick] (4,0) -- (4,\L);

\filldraw[black] (q) circle (1.5pt) node[below=3pt] {$q$};
\filldraw[black] (p) circle (1.5pt) node[below=3pt] {$p$};
\filldraw[black] (qprime) circle (1.5pt) node[below=3pt] {$q'$};

\node at (1,0.3) {$L=0$};
\node at (2.5,\L+0.3) {$L=L^*$};

% Legend
\begin{scope}[shift={(3,-0.2)}]
    \draw[black] (0,2) rectangle (2.3,2.9);
    \fill[blue!20] (0.1,2.6) rectangle (0.4,2.8);
    \node[anchor=west] at (0.5,2.7) {Actual Allocation};

    \fill[red!20, pattern=north east lines, pattern color=red] (0.1,2.3) rectangle (0.4,2.5);
    \node[anchor=west] at (0.5,2.4) {Ideal Allocation};
\end{scope}
\end{tikzpicture}
\end{minipage}
\caption{
% The above figure give visual intuition on how the different kind of trades and what optimal allocation looks like. 
(a) Illustration of the three kind of trades. 
% (1) Overpriced, Arbitrageur and Overshoot, respectively. 
Each vertical line represents a price axis with different positions for \(q, p, q'\) which fully characterizes each type of trade. (1)  \textbf{Overpriced} trade:  \(q'\) is further away from \(p\) than \(q\). (2) \textbf{Arbitrageur} trade: new price $q'$ is closer to the market price. (3) \textbf{Overshoot} trade. 
% \gf{This is not clear how you're counting . Either way, you should label these on the figure itself and move the legend down.}\bt{what about know} 
(b) Liquidity distribution across ticks as a function of price. 
% As the price increases, 
In this overshoot trade example, the optimal strategy is  to allocate the entire budget immediately after the pool price crosses the market price. Formally, all liquidity should be placed within the interval \((p, q')\), where \(p = p_x / p_y\), as price impact is non-positive in this region. Since liquidity can only be minted at discrete tick levels, the optimal placement corresponds to the nearest available ticks—in this case, \((t_{m+2}, t_{m+4})\). This behavior extends to both overpriced and arbitrage trades, depending on tick granularity and pool conditions.} \label{fig:liq}
\label{fig:vertical_combined}
\end{figure}
\subsection{Optimized Investment Strategies by Archetype}
% Our simulation results are consistent with the analytical predictions of Theorem~\ref{theorem:price_impact}: 
% While JIT LPs always earn some fee income, their utility is primarily driven by price impact, as we will show in Section \ref{sec:experiments}. 
% Since fees depend only on the trade size and spot price (i.e., whether a tick is crossed), and not on the market price, the JIT LP’s profitability hinges on its ability to sell overpriced tokens.
% \gf{I don't think the above para belongs here. Let's comment it out? We discuss this later}

Among all archetypes, participating in \emph{overpriced trades} yields the highest utility. These trades allow JIT LPs to sell tokens above market value across the entire chosen range $(a,b)$. 
% \gf{what do you mean by active range?}\bt{Check if now is clear, by active I mean the chosen a,b}

In \emph{overshoot trades}, JIT LPs can profit by targeting only the favorable segment—i.e., by minting liquidity only in the price range where $\Delta x < 0$ and $q < p_x/p_y < q'$. By concentrating liquidity narrowly in this segment, JIT LPs increase their share of fees while avoiding loss-inducing zones.  Figure~\ref{fig:liq}(b) gives an illustrative example; the initial and final price sandwich the market price \( q < p_x/p_y < q'\), so the trade is ``overshooting'' the price. 
% it moves it away from \(p_x/p_y\), but in the other direction, to \(q ' > p_x/p_y\). 
Therefore, ticks between \((q, p_x/p_y\)) belong to the converging part of the trade and have positive price impact, i.e., they are loss-inducing for LPs. The best strategy is to mint liquidity only on feasible ranges inside \((p_x/p_y, q')\), where the price has already crossed the market-implied value and the price impact is non-positive, i.e., favorable to LPs.

% \gf{THis needs a bit more explanation---explain the example in the figure entirely. Where does the price start, where does it end, what options does the LP have, what should they do to maximize profit? It's ok if it's a bit redundant with the caption.}\bt{What about now?}
% is this a toy example or something from real data? Explan in more detail what the figure is showing.}\bt{Its a toy example}

In contrast, engaging in \emph{arbitrageur trades} is generally unfavorable unless the fee revenue is high enough to offset the incurred price impact. The following result formally establishes a sufficient condition under which fees are insufficient to cover price impact:

\begin{proposition}[Fees Insufficient to Offset Price Impact]\label{prop: feeVSil}
   For a transaction with swap parameters $\theta$ (definition in \S\ref{sec:jit-optimization}), suppose without loss of generality that $\Delta x > 0$. For a price range $(a,b)$, if 
   \[
   1+\alpha < \frac{\Delta y_m  \cdot p_y}{\Delta x_m \cdot p_x}  \quad \forall m \in (p(a), p(b)),
   \]
   i.e., if the gained value of $Y$ tokens over the paid value of $X$ tokens is more than the nominal fee rate,  then utility is always negative: \( \mathcal{U}^{(a,b)} < 0\).
\end{proposition}
(Proof in Appendix~\ref{sec:proofprofits})

\section{Empirical Results}
\label{sec:data-analysis}

% We now leverage Algorithm \ref{alg:jit-strategy} to empirically characterize the behavior of optimal JIT strategies under varying market and pool conditions. Through a series of simulations based on real market data, we explore how factors such as trade size, price misalignment, and fee rates influence the selected range and liquidity provision decisions.

% We first present three stylized trading configurations that frequently arise in  CLMMs.  These configurations—referred to as \textit{strategic archetypes}—serve as a conceptual framework for interpreting the actions and  JIT LPs.

% Building on this foundation, we structure our empirical analysis around three key findings, each addressed in its own subsection. Together, these findings expose current inefficiencies in JIT strategies, underscore the central role of price impact in determining JIT profitability, and reveal how optimized JIT behavior could reshape market dynamics.

% \subsection{Experimental Results}

In this section, we empirically evaluate the theoretical insights  
% developed in the previous sections by comparing our optimized JIT strategies against real-world behavior observed on 
using data from Uniswap V3.

\vspace{0.1in}\noindent \textbf{Data Collection.}
We collected data from a USDC/WETH\footnote{0x88e6a0c2ddd26feeb64f039a2c41296fcb3f5640 pool hash} pool over 6 months: %a specific time window of six months of data (\textit{
January-June 2024, which comprises 1,013,147 total transactions. The data was collected from the Allium platform \cite{allium}.
From this data, we extracted all JIT transactions, which were identified by using the same criterion as \cite{JITparadox}, i.e. matching transactions which were sandwiched by a mint and burn transaction by the same provider.
The resulting number of JIT swaps was 6,829.
% and used by authros in  \cite{tang2024game}. %In particular, we have focused on the pool E5:USDC-WETH that has 1428340 events with a fee rare of $0.05\%$. 
% Our primary objective is to identify gaps between theory and practice and to quantify opportunities currently missed by JIT LPs.  
In this section, our results are generated by taking real JIT transactions and simulating our optimized version of JIT investment strategies. 
We assume throughout that a JIT transaction is realized if and only if it was realized in the source data (i.e., we do not simulate different results of bidding auctions among JIT LPs). To simulate a JIT budget, we always assume that the JIT original transaction uses the trader's entire budget. Doing this, we ensure that our simulations do not use more money than  the real execution, so our best results come exclusively from better resource utilization.
% The goal is to compare how suboptimal actual JITs are when compared to the more guided approach. To choose the transactions, 
% \end{itemize}

\subsection{Results}
We observe three main findings from our empirical evaluation: 
\begin{enumerate}
    \item JIT LPs in the pool we studied currently invest suboptimally and could significantly increase their profits.  (\S\ref{sec:missed-price-impact})
    \item JIT LP utility is primarily driven by price impact rather than fees. (\S\ref{sec:price-impact-drives-utility})
    \item If JIT LPs were to optimize their investment strategies, it would help JIT LPs and traders, at the expense of passive LPs. (\S\ref{sec:market-level-effects})
\end{enumerate}

\subsubsection{Suboptimal Investment Strategies}
\label{sec:missed-price-impact}

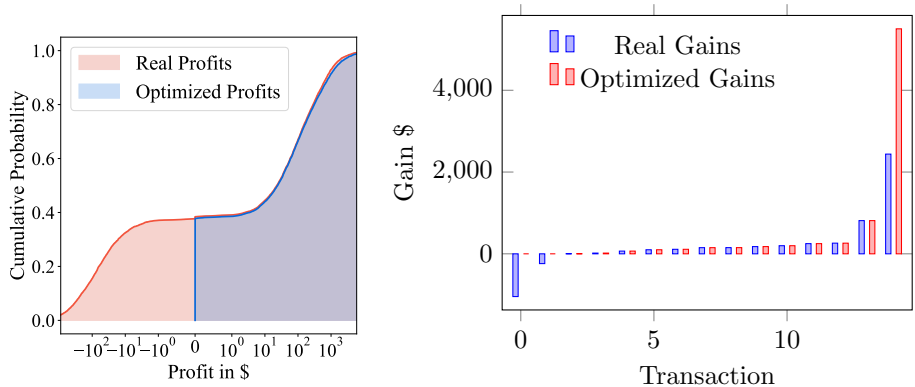
\begin{figure}[t]
    \centering

    % First plot (TikZ)
      \begin{subfigure}[t]{0.40\textwidth}
        \centering
        \includegraphics[width=\linewidth]{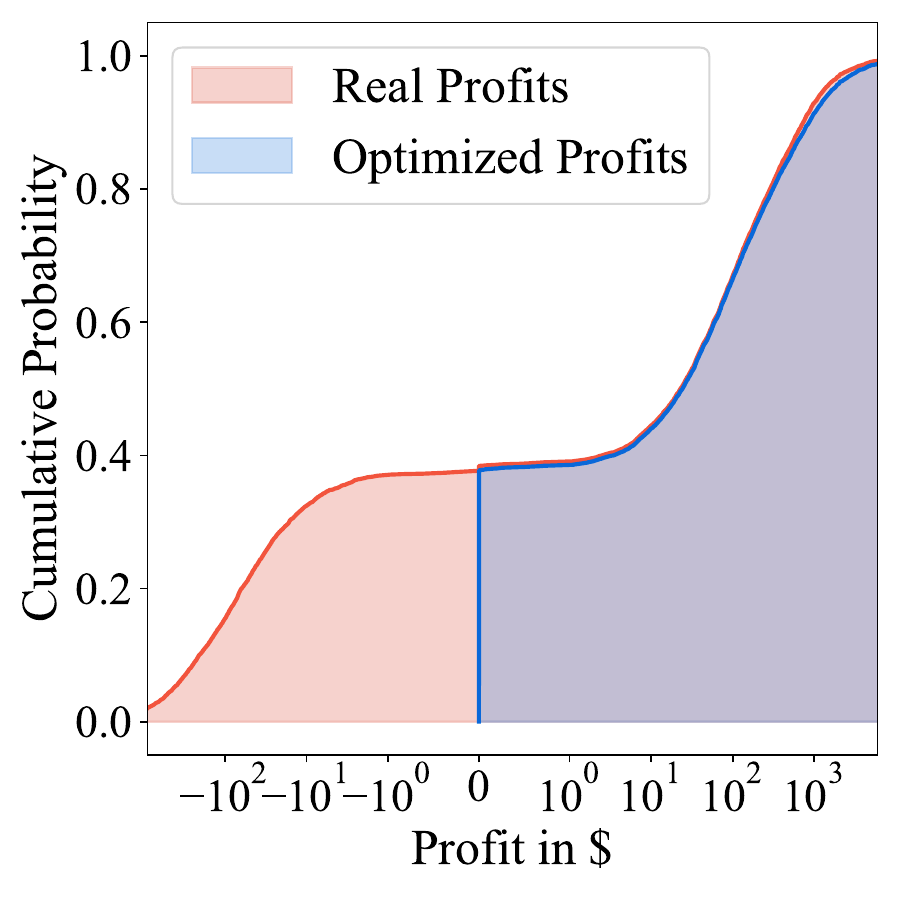}
        \caption{CDF of profit for a real JIT LP vs. optimized strategy. Optimized JIT avoids costly trades with high price impact.}
        \label{fig:profits}
    \end{subfigure}
    \hfill
    \begin{subfigure}[t]{0.58\textwidth}
        \centering
        \begin{tikzpicture}
            \begin{axis}[
                width=\linewidth,
                height=5.5cm,
                xlabel=Transaction,
                ylabel=Gain \$,
                x tick label style={/pgf/number format/1000 sep=},
                enlargelimits=0.05,
                legend style={
                    at={(0.7,0.97)},
                    anchor=north east,
                    draw=none,
                    fill=none
                },
                ybar,
                bar width=2pt,
            ]
            \addplot table [x expr=\coordindex, y=Real Profits, col sep=comma] {data/jitEdgeSample.csv};
            \addplot table [x expr=\coordindex, y=Optimized Profits, col sep=comma] {data/jitEdgeSample.csv};
            \legend{Real Gains, Optimized Gains}
            \end{axis}
        \end{tikzpicture}
        \caption{A sample of 15 transactions comparing absolute profits. While real JITs occasionally perform optimally, losses occur due to participation in losing trades and, at times, suboptimal capital allocation.}
        \label{fig:profitsSample}
    \end{subfigure}
    \caption{Comparison between actual and optimized JIT LP behavior. (a) focuses on aggregate performance, while (b) zooms in on individual transactions.}
    \label{fig:side_by_side_profits}
\end{figure}
Empirical data reveals that real-world JIT LPs underperform significantly relative to their potential gains, as illustrated in Figure~\ref{fig:side_by_side_profits}.
Many JIT agents engage in transactions where 
\begin{figure}
    \centering
    \includegraphics[width=0.6\linewidth]{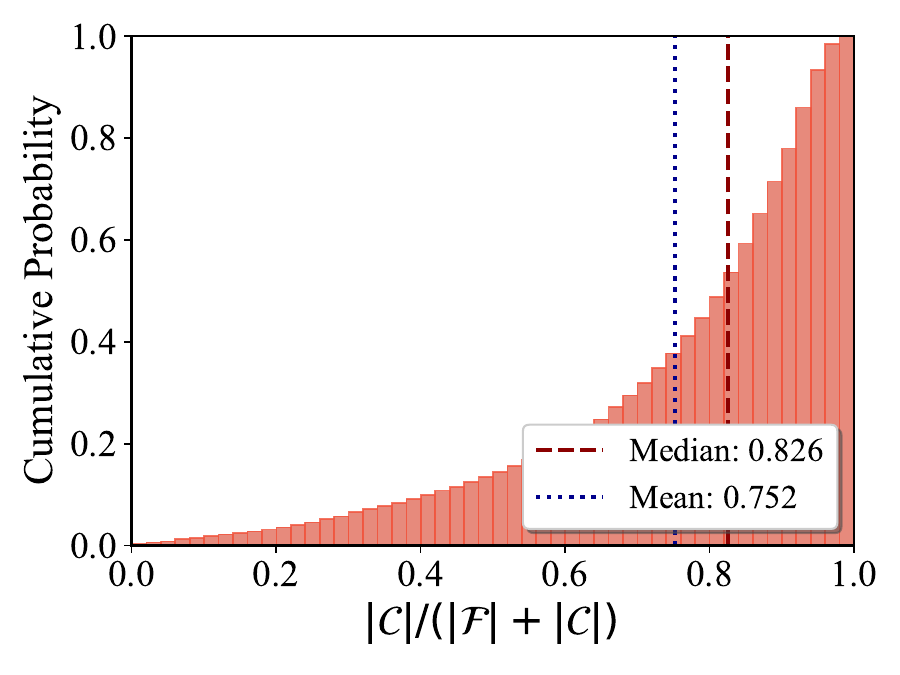}
    \caption{CDF of the proportion of JIT returns resulting from price impact as opposed to fees, i.e., \( \frac{|\mathcal{C}|}{|\mathcal{C}| + |\mathcal{F}|} \). 
    % \gf{Pls change the figure x-axis label to say $|C|/(|C|+|F|)$}\bt{Done} \gf{may need to update the figure?}\bt{Hahahah sorry ! Done now} 
    In most JIT transactions,  price impact has a much greater impact on returns than fees, accounting for $75\%$ of the total returns on average.  }
    \label{fig:perSwapVsPI}
    \vspace{-0.2in}
\end{figure}
 the adverse effects of price movements outweigh the fees earned, resulting in negative net utility (Figure~\ref{fig:side_by_side_profits}(a)). 
In our sample of 6,884 JIT transactions, real JIT LPs could have earned up to 69\% more than their current profit 
% \gf{do you mean 69\% more? "69\% as much" connotes they could have earned 69\% of what they are currently earning, ie less} \bt{Sorry, fixed} 
if they had adopted Algorithm 1.
Our optimized allocation strategy not only helps JIT LPs identify when \emph{not} to invest, but it can also increase returns on already-profitable transactions; for example, Figure~\ref{fig:side_by_side_profits}(b) shows a random sample of 15 JIT transactions, comparing their realized gain in \$ and the estimated gain under Algorithm \ref{alg:jit-strategy}. 
Our optimized strategy never loses money, and it sometimes increases gains relative to the current strategy of JIT traders.
% \bt{Besides, we also observe that, in some transactions, JITs are missing optimal strategies even if they are being profitable, and not earning as much as they could with proper allocation (Figure ~\ref{fig:profitsSample})}.
These observations highlight the importance of modeling price dynamics and underscore the need for strategy optimization beyond mere fee maximization.

% These missteps highlight the importance of modeling price dynamics and stress the need for strategy optimization beyond fee maximization, particularly, in our sample, real JIT LPs performed, when aggregated, 50\% worse then our optimized transactions.

\subsubsection{JIT Utility Is Driven Primarily by Price Impact}
\label{sec:price-impact-drives-utility}

% The main cause of suboptimality appears to be that JIT traders do not account for price impact (or impermanent loss). 
Our analysis shows that the primary driver of JIT profitability is not fee income but price impact (Figures \ref{fig:perSwapVsPI}). Optimized JIT strategies leverage price dislocations to extract value, often outperforming naive approaches that focus solely on fees. 
This illustrates the value of our study relative to prior works  that focus on optimizing fee share \cite{tang2024game}.

% \begin{tikzpicture}\label{img:perSwapFeeVsPI}
% \begin{axis}[
%     width=12cm,
%     height=8cm,
%     xlabel={Trade Volume (USD)},
%     ylabel={IL Relevance (|IL| / (|IL| + |Fees|))},
%     title={IL Relevance vs Trade Volume},
%     grid=both,
%     only marks,
%     enlargelimits=true,
% ]

% % Load the CSV
% \pgfplotstableread[col sep=comma]{data/jit_ilETH-5.csv}\loadedtable

% \addplot+[
%     mark=*,
%     blue,
% ] table[
%     x=total_volume,
%     y expr={
%         abs(\thisrow{il}) / (abs(\thisrow{il}) + abs(\thisrow{fees_values}))
%     }
% ] {\loadedtable};

% \end{axis}
% \end{tikzpicture}

\subsubsection{Market-Level Implications of Optimized JIT Behavior}
\label{sec:market-level-effects}

% Despite being underoptimized, JIT strategies exhibit considerable untapped potential. 
Currently, the participation of JIT LPs in CLMMs is limited, and their share of fee revenue remains negligible when compared to passive LPs; in our dataset, JIT LPs claimed less than 2\% of total fees paid in the system. 
% \begin{wrapfigure}[19]{r}{0.55\textwidth}
%     \centering
%     \includegraphics[width=0.95\linewidth]{imgs/il_cdf.pdf}
%     \caption{CDF of the proportion of JIT returns resulting from price impact as opposed to fees, i.e., \( \frac{|\mathcal{C}|}{|\mathcal{C}| + |\mathcal{F}|} \). 
%     % \gf{Pls change the figure x-axis label to say $|C|/(|C|+|F|)$}\bt{Done} \gf{may need to update the figure?}\bt{Hahahah sorry ! Done now} 
%     In most JIT transactions,  price impact has a much greater impact on returns than fees, accounting for $75\%$ of the total returns on average.  }
%     \label{fig:perSwapVsPI}
%     \vspace{-0.2in}
% \end{wrapfigure}
However, simulations suggest that increased adoption of optimized JIT strategies could substantially reshape the  CLMM  ecosystem. In particular, JIT LPs affect both the distribution of fees among liquidity providers and the execution quality experienced by traders. Figures~\ref{fig:jitBudgetComparison}(a) and~\ref{fig:jitBudgetComparison}(b) illustrate two key market-level effects. As JIT LPs increase their capital allocation, they capture a growing share of the fee revenue, which reduces the earnings of passive LPs by up to 44\%  per trade, when the JIT budget is large (Figure~\ref{fig:jitBudgetComparison}(a)). At the same time, traders benefit from improved execution due to reduced price impact, as the injected liquidity narrows the effective spread and enables more efficient absorption of large orders (Figure~\ref{fig:jitBudgetComparison}(b)).

\begin{figure}[h]
  \centering \setlength{\tabcolsep}{2pt}
  \begin{tabular}{cc}
    % First subplot
    \begin{tikzpicture}
    \begin{axis}[
    	width=0.5\textwidth,
    	xlabel=JIT budget multiplier,
    	ylabel=Fees received by passive LPs (\%),
    	legend style={at={(0.5,-0.2)},anchor=north,legend columns=-1},
        xtick distance=0.5,      
        ytick distance=10,
        yticklabel=\pgfmathprintnumber{\tick}
    ]
    \addplot table [x=JIT Budget Multiplier,
    y expr=\thisrow{Passive LP Fees}/73.27257563792762 * 100, col sep=comma] {data/jitImpact.csv};
    \end{axis}
    \end{tikzpicture}
    &
    % Second subplot
 \begin{tikzpicture}
    \begin{axis}[ width=0.5\textwidth,
        xlabel=JIT budget multiplier,
ylabel={Slippage (\%)},
        y label style={yshift=0em},
        legend style={at={(0.5,-0.1)}, anchor=north, legend columns=-1},
        xtick distance=0.5,
        ytick distance=0.003,
        yticklabel=\pgfmathprintnumber{\tick}
    ]
    \addplot table [
        x = {JIT Budget Multiplier},
        y expr = ((146201.4592813329 - \thisrow{Trader Income})/146201.4592813329*100),
        col sep=comma
    ] {data/jitImpact.csv};
    \end{axis}
\end{tikzpicture}

\\
    (a)  Fees received by passive LPs  & (b) Trader surplus attributable to JIT liquidity
  \end{tabular}
   \caption{Both plots compare the effects of higher JIT budget with a baseline where no JIT engage (x=0). (a) Effect of JIT budget on passive LP fees.  Passive LP fee gains diminish as JIT liquidity increases.  (b) Trader  benefits via reduced slippage with higher JIT budgets. The Y-axis shows the slippage suffered by the trader under different budgets (Recall slippage from \S\ref{ex:slippage}).
   } \label{fig:jitBudgetComparison}
   % \vspace{-0.2in}
\end{figure}
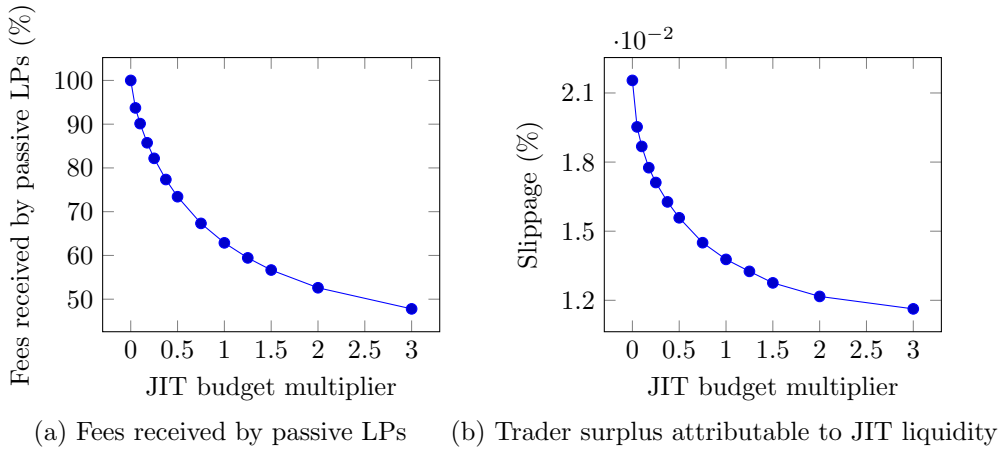
% \gf{This explanation of what you did should come before the result.}

Averaging these values across transactions reveals two trends: the marginal benefit to passive LPs decreases with JIT participation, while traders benefit from reduced slippage. Together, these results indicate that although optimized JIT activity may diminish returns for passive LPs, it can improve capital efficiency and trading experience in CLMMs. This reinforces the view that well-designed JIT strategies can contribute positively to market structure.

\section{Discussion and Implications}
\label{sec:overall}

Our results highlight the strategic advantage of  JIT in CLMMs. Specifically, a JIT LP that enters just before a trade—when the pool price is aligned with the market price—can benefit from the resulting swap-induced price movement. In contrast, passive LPs are exposed to losses when price deviations from the market are corrected via arbitrage.

Importantly, JIT profitability arises only under specific conditions (Proposition~\ref{prop: feeVSil}), often requiring precise timing and positioning (\S\ref{sec:types-of-trades}). Given the competitive nature of bidding for inclusion (e.g., via Flashbot bundles), it is not surprising that real-world JIT activity remains limited.

JIT LPs complement the actions of arbitrageurs. While arbitrageurs restore price alignment after deviations, JIT providers preemptively mitigate such deviations by injecting liquidity in trades with negative price impact \( \mathcal{C} \). In this sense, JIT acts as a frictional force—reducing the magnitude of price dislocations (Figure~\ref{fig:jitBudgetComparison}) that would otherwise create larger arbitrage opportunities. This increased liquidity is beneficial to traders, as it reduces slippage and improves execution quality. However, it also comes at the expense of passive LPs, who face reduced fee income.

\subsection{Real dynamics of providers, arbitrageurs and traders}
A common perception is that arbitrageurs profit primarily at the expense of passive LPs through the  price impact. However, this is not always the case. Consider a model where the market price \( P_t \) evolves stochastically, and the pool price \( Q_t \) follows it with some lag due to noise. Suppose a trade moves the pool price from \( Q_t = P_t \) to \( Q_{t+1} \ne P_{t+1} \), and an arbitrageur subsequently restores alignment at \( Q_{t+2} = P_{t+2} \). If \( P_t = P_{t+1} = P_{t+2} \), the passive LP’s price impact \( \mathcal{C}_n \) is zero, but fees \( \mathcal{F}_n > 0 \) are still earned from both trades.

In this case, the arbitrageur's profit comes not from the LP, but from the noisy trader who executed the initial swap at a worse price—e.g., buying at \( \sqrt{Q_t Q_{t+1}} > Q_t = P_t \). The arbitrageur captures this overpayment at \( \sqrt{Q_{t+1} Q_{t+2}} \). JIT LPs benefit even more from such trades. By sandwiching the swap, they can collect both price impact gains and fees. Moreover, their presence dampens the price displacement, thereby reducing the arbitrageur’s profit and the passive LP’s fee share from the second trade.

To illustrate this, suppose a noisy trade pushes the price to \( Q_{t+1} > Q_t \). A JIT LP that enters before the trade adjusts the effective spot price to \( Q^*_{t+1} \), such that \( Q_{t+1} > Q^*_{t+1} > Q_t \). The arbitrageur, now facing a smaller price deviation, captures less profit: \( \sqrt{Q^*_{t+1} Q_{t+2}} < \sqrt{Q_{t+1} Q_{t+2}} \).

Assuming all arbitrage opportunities are rapidly corrected, the pool price trajectory is driven by market price movements plus correctable noise. This noise does not generate price impact directly; rather, it creates temporary dislocations that can be monetized by arbitrageurs and JIT LPs. In fact, this type of exploitable noise is the basis of Loss-versus-Rebalancing (LVR)~\cite{LVR}, a cost borne by passive LPs even when price impact vanishes.

JIT LPs exploit this noise by positioning liquidity around anticipated dislocations. Their interventions reduce price volatility, preempt arbitrage, and extract part of the fees that would otherwise accrue to passive LPs or arbitrageurs. If arbitrage profits stem from noisy trades rather than passive LPs, JIT LPs can be seen as first movers that mitigate price deviations while capturing a portion of this surplus.

 \section{Related Work}
\label{sec:related}

Despite a growing body of work analyzing LP incentives in AMMs, including recent theoretical models of JIT LPs~\cite{JIT,JITparadox}, prior efforts either address market-level dynamics under informational asymmetry~\cite{JITparadox} or empirically characterize JIT activity~\cite{JIT}. To our knowledge, no existing work models the transaction-level, optimization-based behavior of JIT LPs in CLMMs such as Uniswap V3. We fill this gap by formally modeling the per-transaction utility of JIT LPs based on price impact (\S\ref{sec:imperloss}) and fee accrual (\S\ref{sec:fees}), analyzing how optimal strategies vary with trade direction and liquidity depth, and comparing predictions with observed on-chain behavior.

Complementing our work, Cartea et al.~\cite{cartea2024decentralized} develop a continuous-time model for strategic liquidity provision in CLMMs. They derive optimal dynamic strategies for LPs adjusting liquidity ranges over time, accounting for fee income, predictable loss, and concentration risk. Their results show how exchange rate drift and rebalancing costs shape behavior. Unlike our discrete, transaction-level model where JIT LPs react to individual swaps, their focus is on long-term LPs managing ongoing exposure. Still, their analysis highlights the need to model fee income and adverse selection in CLMMs.

Tang et al.~\cite{tang2024game} present a game-theoretic model of liquidity provision in CLMMs, focusing on interactions between passive LPs. Their static model captures how budget constraints and price tick granularity affect equilibrium behavior, identifying a unique Nash equilibrium with a water-filling strategy. By comparing theory to on-chain behavior, they show LPs in risky pools deviate more from optimality than those in stable pools. 
However, they do not model or consider the effects of JIT LPs.
% Though our model targets reactive, single-shot optimization by JIT LPs, their findings underscore the gap between theoretical and real-world LP behavior, which we also address via simulation and data-driven modeling.

In this work, we analyze price impact as a key metric. Price impact and LVR coincide for individual trades~\cite{LVRvsILI,LVRvsILII}. While LVR matters for passive LPs over time, JIT LPs primarily care about price impact.

Recent studies explore how informational asymmetries and adversarial timing affect CLMM dynamics~\cite{daian2020flash,qin2022quantifying,JITparadox,heimbach2022risks}. Capponi et al.\cite{JITparadox} model adverse selection and value extraction by informed traders. Qin et al.\cite{qin2022quantifying} and Daian et al.~\cite{daian2020flash} examine MEV, showing how latency arbitrage and frontrunning harm LPs. While these works address systemic risks, we focus on microstructure-level decisions by JIT LPs under fee competition and price impact. Our formal framework captures their transaction-specific optimization, and shows how misjudging price impact can erode profits—even for strategically-timed liquidity.

\section{Conclusion}
\label{sec:concl}

Emerging tools such as  CLMMs, Flashbots, and Uniswap V4 Hooks are reshaping DeFi by enabling JIT  liquidity providers to improve capital efficiency, manage risk, and align strategies with real-time on-chain dynamics. Yet, empirical evidence shows that most JIT agents fail to exploit profitable opportunities and often incur losses due to simplistic strategies.

This work demonstrates that significant gains are achievable via optimization-based approaches that explicitly weigh fee accrual against price impact. Our transaction-level framework guides JIT agents toward utility-maximizing strategies that allocate capital more effectively and reduce adverse outcomes such as slippage or loss.

More broadly, the results support a shift toward reactive, transaction-driven market models. By analyzing individual trades, rather than long-term aggregates, we obtain sharper insights into the microstructure of AMMs and the role of JIT behavior.

Future research should explore extensions of AMM models that account for JIT strategies and   protocols that safeguard fairness for passive LPs. Closing the loop between protocol design and agent strategy is key to building more adaptive and sustainable DeFi ecosystems.

\section*{Acknowledgments}
This work was partially funded by CAPES; by FAPERJ under grant E-26/204.268/2024; and by the Conselho Nacional de Desenvolvimento Científico e Tecnológico (CNPq), Brazil, under grants 403601/2023-1 and 315106/2023-9. 
It was also supported in part by the National Science Foundation under grant CNS-2325478, as well as by the Initiative for Cryptocurrencies and Contracts (IC3) and the CyLab Secure Blockchain Initiative, together with their respective industry sponsors.

\bibliography{main}

% \newpage 

\appendix

\appendix

\section{Notation}\label{sec:notation}

Table~\ref{tab:notation} summarizes the notation used throughout this work.

\begin{table}[htb]
\centering
\renewcommand{\arraystretch}{1.3}
\caption{Summary of Notation}
\begin{tabular}{ll}
\hline
\textbf{Symbol} & \textbf{Description} \\ \hline
$q$ & Pool price (amount of token $Y$ per unit of token $X$) \\
$q'$ & New pool price after a swap \\
$q^*$ & New pool price after a swap with no JIT participation \\
$p_x$, $p_y$ & Market prices of tokens X and Y in USD \\
$x_n(q)$, $y_n(q)$ & Amount of tokens X and Y held by LP $n$ at pool price $q$ \\
$(L, a, b)$ & A liquidity position: liquidity $L$ between ticks $a$ and $b$ \\
% $\pos_n$ & Set of liquidity positions chosen by LP $n$ \\
$t_m$ & Tick price: $t_m = 1.0001^{m\cdot \tau}$ \\
% $\hat{q}$ & Clamped price: $\hat{q} = \max(a, \min(b, q))$ \\
$\V_{(a,b)}$ & Dollar price per unit of liquidity in range $(a,b)$ \\
$\mathcal{C}_n$ & Absolute price impact for LP $n$ \\
$\mathbf{PI}$ & Relative price impact  (ratio of loss to initial value) \\
$\mathcal{F}_n$ & Total fee reward for LP $n$ \\
$P_{n,m}$ & Liquidity provided by passive LP $n$ active at tick $m$ \\
$P_m$ & Pasive liquidity active at tick $m$ \\
$\delta_m$ & Fees paid in tick $m$ \\
% $y^*(m)$ & Available amount of token Y to be swapped at tick $m$ \\
% $r(m)$ & Remaining amount of Y to be swapped after crossing $m$ ticks \\
$v_j^*$ & Gas cost or bid paid by JIT LP \\
\hline
$\mathcal{U}$ & Utility of JIT LP \\
$s$ & Strategy of JIT LP $(L,a,b,v)$ \\
$\rho$ & Budget of JIT LP \\
$\mathcal{S}$ & Strategy space of JIT LP \\
\( ([N], \mathcal{U}, \mathcal{S}, \theta, \rho) \) & JIT LP optimization problem \\
\hline
\end{tabular}
\label{tab:notation}
\end{table}

% \section{Basics of AMM}\label{sec:basics}

% ----------------------------------
% Theorem 1
% ----------------------------------

\section{Proof of Theorem~\ref{theorem:price_impact}} 
\label{sec:priceimpactproof}
\begin{reptheorem}[Threshold condition for price impact \ref{theorem:price_impact}]
Let       $\hat{q} \triangleq \max\{a, \min\{b, q\}\}$, and similarly for $\hat{q}'$. 
% \gf{I would move this definition before the bulleted list.} 
The change in pool price relates to price impact as follows:
\begin{itemize}
    \item Case 1: \( q' < q \). Then, \( \mathcal{C} \le 0 \) if and only if
    \begin{equation}
    \frac{1}{\hat{q}} \left( \frac{p_x}{p_y} \right)^2 \ge \hat{q}'. \label{ineq1app}
    \end{equation}
    
    \item Case 2: \( q < q' \). Then, \( \mathcal{C} \le 0 \) if and only if
    \begin{equation}
    \frac{1}{\hat{q}} \left( \frac{p_x}{p_y} \right)^2 \le \hat{q}'. \label{ineq2app}
    \end{equation}
\end{itemize}
\end{reptheorem}

\begin{proof}
We prove the case \( q < q' \); the argument for \( q' < q \) is symmetric. In addition,  we assume $\hat{q} \neq \hat{q}' $ as the case $\hat{q} = \hat{q}' $ yields  $ \mathbf{PI} = 0 $.

% Theorem~\ref{theorem1}.

We aim to show \( \mathbf{PI} \le 0 \) if and only if~\eqref{ineq2app} holds. Since \( p_x = p_x' \) and \( p_y = p_y' \), we can use the simplified expression for  price impact:
\[
\mathbf{PI} = 1-\frac{V_\textrm{withdraw}(\bm{p})}{V_\textrm{mint}(\bm{p})}= 1 - \frac{p_x x_n(q') + p_y y_n(q')}{p_x x_n(q) + p_y y_n(q)}.
\]
This expression is non-positive if and only if the final value (numerator) is greater than or equal to the initial value (denominator), i.e.,
\[
p_x x_n(q') + p_y y_n(q') \ge p_x x_n(q) + p_y y_n(q).
\]

Substituting the definitions of \( x_n(q) \) and \( y_n(q) \) from \eqref{eq:position-holdings}, and focusing on a single position \((L,a,b)\), we obtain:
\[
p_x \left( \frac{1}{\sqrt{\hat{q}'}} - \frac{1}{\sqrt{b}} \right) + p_y \left( \sqrt{\hat{q}'} - \sqrt{a} \right) \ge p_x \left( \frac{1}{\sqrt{\hat{q}}} - \frac{1}{\sqrt{b}} \right) + p_y \left( \sqrt{\hat{q}} - \sqrt{a} \right).
\]
Eliminating common terms on both sides leads to:
\[
p_x \left( \frac{1}{\sqrt{\hat{q}'}} - \frac{1}{\sqrt{\hat{q}}} \right) \ge p_y \left( \sqrt{\hat{q}} - \sqrt{\hat{q}'} \right).
\]
Now, dividing both sides by \( \sqrt{\hat{q}} - \sqrt{\hat{q}'} > 0 \) (since \( q < q' \Rightarrow \hat{q} \leq \hat{q}' \)), and recalling that we assume $\hat{q} \neq \hat{q}'$,  we get the desired inequality:
\[
\frac{p_x}{p_y} \le \sqrt{\hat{q} \hat{q}'}.
\]
The reverse inequality gives \( \mathbf{PI} > 0 \). % completing the proof.

 The  remainder of the proof follows from the above,  by swapping the signs. \dsm{maybe improve here}
\end{proof}

% %%%%%%%%%%%%%%%%%%%
% Corollary 1
%%%%%%%%%%%%%%%%%%%%%%%

\subsection{Proof of Corollary~\ref{coro2}}

\begin{repcorollary}[Gains under diverging prices  and losses under  converging prices \ref{coro2}]
When initial and final prices do not cross $p_x/p_y$, gains and losses are determined by the direction of the price movement:
\begin{itemize}
    \item Gain under diverging prices: 
\( \mathcal{C} \le 0 \) if \( q' < q \leq  {p_x}/{p_y} \)  or \( q' > q \ge {p_x}/{p_y} \), i.e., if  a trade moves the AMM price \emph{away} from the external market price, the LP experiences gains.
\item Loss under converging prices:  \( \mathcal{C} \ge 0 \) if \( q < q' < {p_x}/{p_y} \)  or \( q > q' > {p_x}/{p_y} \),  i.e., if  a trade moves the AMM price \emph{towards} the external market price, the LP experiences losses.
\end{itemize}
\end{repcorollary}
\begin{proof}
    If $q < q' < {p_x}/{p_y}$ then:

    \begin{equation*}
    \begin{split}
        &\left(\frac{p_x}{p_y}\right)^2 > q q' \implies \frac{1}{\hat{q}}\left(\frac{p_x}{p_y}\right)^2 > \hat{q}'
    \end{split}
    \end{equation*}

\noindent
    and, from Theorem~\ref{theorem:price_impact}, we have  $\mathbf{PI} \ge 0$. For the case $q > q' > {p_x}/{p_y}$ 
    \begin{equation*}
    \begin{split}
        &\left(\frac{p_x}{p_y}\right)^2 < q q' \implies \frac{1}{\hat{q}}\left(\frac{p_x}{p_y}\right)^2 < \hat{q}'
    \end{split}
    \end{equation*}

    \noindent
    and $\mathbf{PI} \ge 0$.
    
\end{proof}

\section{Proof of Lemma~\ref{lemma:limit_price_impact}} \label{sec:prooflemalimitingbehavior}
% ----------------------------------
% Lemma 2
% ----------------------------------
\begin{replemma}[Limiting behavior of price impact \ref{lemma:limit_price_impact}]
    Let \( \Delta x \) be the quantity of token \( X \) exchanged in a trade. Consider \( q \in (t_m, t_{m+1}]\) the current pool price and consider \( ( L, a,b)\) to be the single position of some provider $n$ such that $q \in (a,b)$, in words, consider \( L \) to be all the liquidity minted before the transaction by some provider $n$. 

    Then:
\[
\lim_{L \to \infty} \mathcal{C} = \Delta x \cdot (p_y \cdot q - p_x), \quad \text{and} \quad \lim_{L \to 0} \mathcal{C} = 0.
\]
\end{replemma}

\begin{proof}\label{proof: Lemma6}
    From \eqref{eq:priceimpact} and \eqref{eq:position-holdings} we have    \begin{equation*}
        \mathcal{C} = -p_x\left(\frac{1}{\sqrt{q}} - \frac{1}{\sqrt{q'}}\right) L - p_y\left(\sqrt{q} - \sqrt{q'} \right) L.
    \end{equation*}
Suppose now that $\Delta x$ is given and that there is strictly positive liquidity $P_m$ at tick \( m \). It is straightforward to see that if no \( L \rightarrow 0 \)), then \( \mathcal{C} \rightarrow 0 \), since there is no participation in the trade.
\noindent
Alternatively, from {Lemma~\ref{lemma:continuity-monotonicity}}, we know that \( q' \) is monotonically decreasing in \( L \). Moreover, since \( q' \to q \) as \( L \to \infty \), we conclude that in the limit of infinite liquidity, the final price converges to the initial price.
\[
\lim_{L\rightarrow \infty} \mathcal{C} = \lim_{L\rightarrow \infty} -p_x\left(\frac{1}{\sqrt{q}} - \frac{1}{\sqrt{q'}}\right) L - p_y\left(\sqrt{q} - \sqrt{q'} \right)L
\]

\noindent
For very large \(L\) we know that $q'$ and $q$ will be very close, therefore we can safely assume that at some point they will be in the same tick \(m\). Hence, we can use Equation~\eqref{eq:position-holdings} to rewrite the cost entirely in terms of \( L \):

\[
\lim_{L\rightarrow \infty} \mathcal{C} = \lim_{L\rightarrow \infty} -p_x\left(\frac{\Delta x}{L+P_m}\right) L - p_y\left(\frac{\Delta y}{L+P_m} \right) L.
\]

\noindent
Using the relations from Equation~\eqref{eq:position-holdings}, we can express \( \sqrt{q'} \) in terms of \( \Delta x \), and substitute it into the expression for \( \Delta y \), thus writing \( \Delta y \) as a function of \( \Delta x, P, L \). Substituting this into the cost equation yields:

\[
\lim_{L\rightarrow \infty} \mathcal{C} = \lim_{L\rightarrow \infty} -p_x\left(\frac{L}{L+P_m}\right) \Delta x + p_y\left(\Delta x q \frac{L+P_m}{\Delta x \sqrt{q} + L + P_m} \right) \left(\frac{L}{L+P_m}\right).
\]

\noindent
Since \( \lim_{L \to \infty} \frac{L}{L+P_m} = 1 \) and \( \lim_{L \to \infty} \frac{L+P_m}{\Delta x \sqrt{q} + L + P_m} = 1 \), we obtain:
\[
\lim_{L\rightarrow \infty} \mathcal{C} = -p_x \Delta x + p_y \Delta x q = \Delta x(p_y q - p_x),
\]
So, we conclude:

\[
\lim_{L \to 0} \mathcal{C} = 0, \qquad \lim_{L \to \infty} \mathcal{C} = \Delta x(p_y q - p_x).
\]
\end{proof}
The intuition behind Lemma~\ref{lemma:limit_price_impact} is that when there is no liquidity participation, the JIT LP incurs no price impact. Conversely, if an enormous amount of liquidity is injected, then the final price \( q' \) converges to the initial price \( q \) as \( L \to \infty \), and all of the price impact reduces to the net gain or loss from exchanging token \( X \) for token \( Y \) at the fixed rate \( q \).

% ----------------------------------
% Theorem 2
% ----------------------------------

\section{Proof of Theorem~\ref{theorem:utility_opt}}\label{app:last_theorem}

In this section we show the proof of \textbf{Theorem \ref{theorem:utility_opt}}. We begin by first showing two necessary results \textbf{lemmas~\ref{lemma:continuity-monotonicity} \ref{lemma:limit-fee}} and the making the demonstration of the theorem.

\begin{lemma}\label{sec:limitfeeproof}
Consider an LP $n$ with a single position $(L,a,b)$, such that \(q \in (a, b)\). Let there be a swap that pays \( \Delta x > 0 \) $X$ tokens. The total fee earned by the provider is given by:
\[
\mathcal{F}(L) = \sum_{\forall m| t_m \in (a,b)} \delta_m \cdot \frac{L}{L + {\tilde{P}_{n,m}}},
\]
where \( \tilde{P}_{n,m} = \sum_{i \not= n } P_{i,m}\) is the total liquidity in tick $m$ without $ L = P_{n,m} $. In addition, \( \mathcal{F}(L) \) is continuous in \( L \) and is bounded above by \( \alpha  \cdot \Delta x \),
\begin{equation}
    \lim_{L \rightarrow \infty} \mathcal{F}(L) \leq \alpha \cdot \Delta x. 
\end{equation}
\end{lemma}

\begin{proof}

\noindent
Without loss of generality, assume that the trade is selling token X, $\Delta x > 0$ (see Remark~\ref{rmk:direction} if the trade is selling token Y).

Now, define the \emph{fee function} as
\[
\delta_m \triangleq \Delta x_m \cdot \alpha \cdot p_x,
\]
where $\Delta x_m$ is the amount of token $X$ exchanged at tick $m$, $\alpha$ is the pool’s fee rate, and $p_x$ is the dollar value of one unit of token $X$.  
Recall that  $P_m$ denotes the total liquidity provided by passive LPs at tick $m$. Denote $\tilde{P}_{m,n} := \sum_{i \not =n} P_{i,m} $  as all the liquidity \textbf{except} the liquidity $P_{n,m} = L \in \mathbb{R}^+ $ of provider $n$. Substituting this into the expression for fee distribution, we obtain:
\begin{equation*}
\begin{split}
&\mathcal{F}_{m} = \Delta x_m \alpha \frac{P_{n,m}}{\tilde{P}_{n,m} + P_{n,m}} p_x \\
&\mathcal{F} = \sum_{m|t_m \in (a,b)} \mathcal{F}_{m}
\end{split}
\end{equation*}

\noindent
Now let $\Delta x_m$ be such that, 
\begin{equation*}
 \Delta x_m =
    (\tilde{P}_{n,m} + P_{n,m})\left(\frac{1}{\sqrt{\hat{q}'}} - \frac{1}{\sqrt{a_m}}\right).
\end{equation*}

\noindent
Notice that $q' \rightarrow q$ as $P_{n,m} \rightarrow \infty$. Since $\Delta x_m$ is continuous in $q'$ and $q'$ is continuous in $P_m$ by {Lemma~\ref{lemma:continuity-monotonicity}}, it follows that $\Delta x_m$ is also continuous in $P_m$ and hence in $P_{n,m}$. Finally, since $\mathcal{F}$ is defined as a sum of continuous functions, it is itself continuous in $P_{n,m}$.

\end{proof}

% ----------------------------------
% Lemma 1
% ----------------------------------

\begin{replemma}[\ref{lemma:continuity-monotonicity}]
\label{sec:proofcontinuitymonotonicity}
Let \(q\) be the current price in a CLMM, where $q \in (t_{m-1}, t_{m}]$. 
Suppose we have a target trade that pays $\Delta x$ $X$ tokens to the pool and asks for $Y$ tokens, which can be supported by the CLMM. 
Immediately before the trade, a JIT LP adds liquidity $L \in [0, \infty)$ to price range $(a, b)$ where $a \le t_{m-1} < t_{m} \le b$. 
Let $q' (< q)$ be the resulting price after the trade as a function of $L$. 
Then, $q'$ is continuous, strictly decreasing in $L$, and $\lim_{L \to \infty} q' = q$
    \begin{equation}
    \lim_{P_\lfloor q\rfloor \rightarrow \infty} q' = q.   
     \end{equation}

\end{replemma}

\begin{proof}
Let $P_k$ denote the aggregated liquidity in atomic range $(t_{k-1}, t_{k}]$ for each $k \in [M]$. Suppose $a = t_{\mu-1}$ and $b = t_\nu$ where $\mu \le m \le \nu$. 
% To prove the statement, we only need to show that $L$ is continuous and strictly decreasing in $q'$, where $q'$ is the target ending price of the $\Delta x$ trade, and $L$ is the amount of liquidity in $(a, b)$ needed from the JIT LP to reach $q'$. 

Let $q' \in (t_{\ell-1}, t_\ell]$ where $\ell \le m$. By the CLMM policy, we have 
\begin{align*}
    \Delta x 
    &= \sum_{k=\ell}^{m} \big(P_k + L \cdot \mathbb{I}[\mu \le k \le \nu]\big) \left( \frac{1}{\sqrt{\max\{q', t_{k-1}\}}} - \frac{1}{\sqrt{\min\{q, t_{k}\}}} \right) \\ 
    &= \sum_{k=1}^{m} \big(P_k + L \cdot \mathbb{I}[\mu \le k \le \nu]\big) \max\left\{\frac{1}{\sqrt{\max\{q', t_{k-1}\}}} - \frac{1}{\sqrt{\min\{q, t_{k}\}}}, 0 \right\}.
\end{align*}

Equivalently, we can denote 
\[
    \phi_k(L) \triangleq P_k + L \cdot \mathbb{I}[\mu \le k \le \nu], \qquad 
    \psi_k(q') \triangleq \max\left\{\frac{1}{\sqrt{\max\{q', t_{k-1}\}}} - \frac{1}{\sqrt{\min\{q, t_{k}\}}}, 0 \right\}.
\]

We have the following observations:
\begin{enumerate} 
    \item \label{item:phi1} $\phi_m(L) = P_m + L$; 
    \item \label{item:phi2} $\phi_k(L)$ is monotonically non-decreasing and continuous in $L$ for all $k \in [m-1]$; 
    \item \label{item:psi1} $\psi_m(q') > 0$;
    \item \label{item:psi2} $\psi_k(q')$ is monotonically non-increasing and continuous in $q'$ for all $k \in [m]$;
    \item $\sum_{k=1}^m \phi_k(L) \psi_k(q') = \Delta x$ is constant. 
\end{enumerate}

Therefore, as $L$ increases, $q'$ must not increase by \ref{item:phi2} and \ref{item:psi2}. $q'$ must not stay constant, either, by \ref{item:phi1} and \ref{item:psi1}, which implies the strict monotonicity of $q'$ in $L$. 
The continuity of $q'$ in $L$ originates from the continuity of each $\phi_k$ and $\psi_k$. 

Finally, as $L \to \infty$, if for some $k < m$, $\psi_k(q') \neq 0$, then $q' \le t_k$ and $\psi_m(q') = t_m^{-1/2} - q^{-1/2}$ is a positive constant. This contradicts to $\Delta x \ge \phi_m(L) \psi_m(q') = \infty$. 

Hence we must have $\psi_k(q') = 0$ for all $k < m$, which implies $t_{m-1} \le q' < q$. Consequently, 
\begin{alignat}{3}
     \Delta x 
    &= (P_m + L) \left( \frac{1}{\sqrt{q'}} - \frac{1}{\sqrt{q}} \right) \quad & \Longleftrightarrow \\
    q' &= \left( \frac{\Delta x}{P_m + L} + \frac{1}{\sqrt{q}} \right)^{-2} &  \Longrightarrow \\
    \lim_{L\to\infty} q' &= \left(\frac{1}{\sqrt{q}}\right)^{-2} = q. 
\end{alignat}

\end{proof}

\label{sec:prooflasttheo}

\begin{reptheorem}[\ref{theorem:utility_opt}]
    The utility function \( \mathcal{U} \), as defined in~\eqref{tickUtility}, attains a global maximum over the strategy space \( \mathcal{S} \). That is, there exists an optimal strategy \( s^* \in \mathcal{S} \) such that \( \mathcal{U}(s^*) \ge \mathcal{U}(s) \) for all \( s \in \mathcal{S} \).
\end{reptheorem}

\begin{proof}
Fix a specific swap $(\Delta x, q, P, p_x, p_y,\alpha)$ and let $q^\ast$ denote the pool price after the swap assuming no JIT participation. By construction (\S\ref{sec:back-jit}), any rational JIT LP will choose a price range $(a, b) \subseteq (\lfloor q \rfloor_M, \lceil q^\ast \rceil_M)$ (or its reverse if $q^\ast < q$). Since tick prices are discrete, the set of such price intervals is finite. Fix a pair $(a, b)$ and consider all feasible liquidity amounts $L$ such that the total cost does not exceed the JIT's budget $\rho$, i.e.,
\[
L \cdot \V_{(a,b)} + v \leq \rho.
\]
\noindent
The utility function is:
\[
 \mathcal{U}^{(a,b)}(L) = \mathcal{F}^{(a,b)}(L) - \mathcal{C}^{(a,b)}(L) - v^\ast,
\]
where $v^\ast$ is fixed in the simplified setting.

From Lemma 1, $q'$ is continuous in $L$, and from {Lemma~\ref{lemma:limit_price_impact}} and {Lemma~\ref{lemma:limit-fee}}, both $\mathcal{C}$ and $\mathcal{F}$ are continuous in $L$. In particular:
\begin{itemize}
    \item $\mathcal{C}(L) \to 0$ as $L \to 0$, since no liquidity implies no price impact.
    \item $\mathcal{C}(L) \to \Delta x (p_y q - p_x)$ as $L \to \infty$, capturing the full cost of slippage.
    \item $\mathcal{F}(L)$ is bounded above by the total fee $\delta$ paid by the swap and is continuous due to its dependence on $L$ and the fee allocation formula.
\end{itemize}

\noindent
Therefore, for each fixed $(a, b)$, $\mathcal{U}^{(a,b)}(L)$ is continuous on the compact interval $
\left[0, {(\rho - v)}/{\V_{(a,b)}}\right]$, 
and   achieves a maximum by the extreme value theorem.

 Since the set of valid $(a, b)$ is finite and $\mathcal{U}$ has a maximum for each fixed $(a, b)$, it follows that a global maximum over the entire strategy space $\mathcal{S}$ exists.

\end{proof}

\section{Proof of Proposition~\ref{prop: feeVSil}} \label{sec:proofprofits}

\begin{repproposition}[\ref{prop: feeVSil}]
   For a transaction with parameters $\theta$ (assuming without loss of generality that $\Delta x > 0$) and a price range $(a,b)$, if 
   \[
   \frac{p_x}{p_y} < \frac{\Delta y_m}{\Delta x_m (1+\alpha)} \quad \forall m \in (p(a), p(b)),
   \]
   then \(\mathcal{F}^{(a,b)} < \mathcal{C}^{(a,b)} \Rightarrow \mathcal{U}^{(a,b)} < 0\).
\end{repproposition}
\begin{proof} 
%
% Recall that   \( \mathcal{F} \) is the received fee, hence \( \mathcal{F} \ge 0 \). So if we want to show the sign of \( \mathcal{U}\) we must only understanding when \( \mathcal{C} > \mathcal{F} \). 
%
Recall that \( \mathcal{F} \) denotes the received fee and is therefore non-negative, i.e., \( \mathcal{F} \ge 0 \). To determine the sign of \( \mathcal{U} \), it suffices to understand when \( \mathcal{C} > \mathcal{F} \).

Consider then a specific position \( (a,b) \) and the tick-versions of \( \mathcal{C} \) and \( \mathcal{F}\). Note that
\[ \mathcal{F}^{(a,b)}_m < \mathcal{C}^{(a,b)}_m \; \forall m \implies \mathcal{U}^{(a,b)} > 0.
\]

From Appendix~\ref{proof: Lemma6}, %we  can express $\mathcal{C}_m$ as
$$\mathcal{C}^{(a,b)}_m = -p_x\left(\frac{L}{L+P_m}\right) \Delta x_m + p_y
 \Delta y_m \left(\frac{L}{L+P_m}\right).$$
Then,
\[
 \Delta x_m \alpha \frac{L}{L + P_m}p_x <  -p_x\left(\frac{L}{L+P_m}\right) \Delta x_m + p_y
 \Delta y_m \left(\frac{L}{L+P_m}\right)
\]
Simplifying $\frac{L}{L+P_m}$ and \(\Delta x_m\)
\[
  \alpha p_x <  -p_x + p_y\frac{\Delta y_m}{\Delta x_m} \implies 
  \frac{p_x}{p_y}(1+ \alpha) <  \frac{\Delta y_m}{\Delta x_m}.
\]
Therefore,
\[
\frac{p_x}{p_y} <  \frac{\Delta y_m}{\Delta x_m(1+ \alpha)}.
\]
The above inequality is a sufficient condition for $\mathcal{U}^{(a,b)} > 0$, which concludes the proof.
\end{proof}
\section{Utility Function Simulation Visualization}\label{app:images}

To visualize the impact of strategic positioning in JIT provisioning, we simulate optimal strategies for each of the three archetypes introduced earlier: \emph{arbitrageur}, \emph{overpriced}, and \emph{overshoot} trades. For each case, we select a representative large trade (with $\Delta x > 0$, implying $q' < q$) from Section~\ref{sec:data-analysis} that induces significant price movement. This allows us to evaluate how utility varies with budget and tick placement.

Figure~\ref{fig:arbitra} presents the utility curves for an \textit{arbitrageur trade}, where all JIT strategies result in negative utility despite fee collection, due to adverse price impact. In contrast, Figure~\ref{fig:overprice} shows an \textit{overpriced trade}, in which all positions yield positive utility—consistent with Corollary~\ref{coro2}—and the optimal strategy depends on the available budget. Finally, Figure~\ref{fig:oversho} illustrates an \textit{overshoot trade}, where certain positions only become profitable after the price crosses the market exchange rate $p_x/p_y$, with the highest utility achieved in regions where $\mathcal{C} < 0$.

\begin{figure}
\centering
\begin{tikzpicture}
\begin{axis}[
    xlabel={Budget in \$},
    xmin=0,
    xmax= 1000000,
    ylabel={Utility in \$},
    legend style={at={(0.3,0.4)}, anchor=north east},
    grid=both,
    width=8cm,
    height=8cm,
    line width=0.5pt 
]

\addplot [domain=0:1000000, dashed, black, thick] {0};
\addlegendentry{No participation}

\addplot [mark=none, color=blue] table [x=budget, y={-2; -1}, col sep=comma] {data/utility_arbitrage_trade.csv};
\addlegendentry{$-2, -1$}

\addplot [mark=none, color=red] table [x=budget, y={-2; 0}, col sep=comma] {data/utility_arbitrage_trade.csv};
\addlegendentry{$-2,  \phantom{-}0$}

\addplot [mark=none, color=cyan] table [x=budget, y={-2; +1}, col sep=comma] {data/utility_arbitrage_trade.csv};
\addlegendentry{$-2, +1$}

\addplot [mark=none, color=magenta] table [x=budget, y={-1; 0}, col sep=comma] {data/utility_arbitrage_trade.csv};
\addlegendentry{$-1,  \phantom{-}0$}

\addplot [mark=none, color=green] table [x=budget, y={-1; +1}, col sep=comma] {data/utility_arbitrage_trade.csv};
\addlegendentry{$-1, +1$}

\addplot [mark=none, color=brown] table [x=budget, y={0; +1}, col sep=comma] {data/utility_arbitrage_trade.csv};
\addlegendentry{$ \phantom{-}0, +1$}

\end{axis}
\end{tikzpicture}
\caption{\textit{Arbitrageur trade:}  Utility curves for different choices of $(a,b)$. The legend indicates how many ticks away the position endpoints are relative to $\lfloor q \rfloor$. For example, the blue line labeled $-2, -1$ corresponds to the position $(t_{\lfloor q \rfloor - 2}, t_{\lfloor q \rfloor - 1})$. Despite the presence of fees, all strategies yield negative utility in this scenario, as the price impact is too severe for the sandwich strategy to be profitable.} \label{fig:arbitra}
\end{figure}
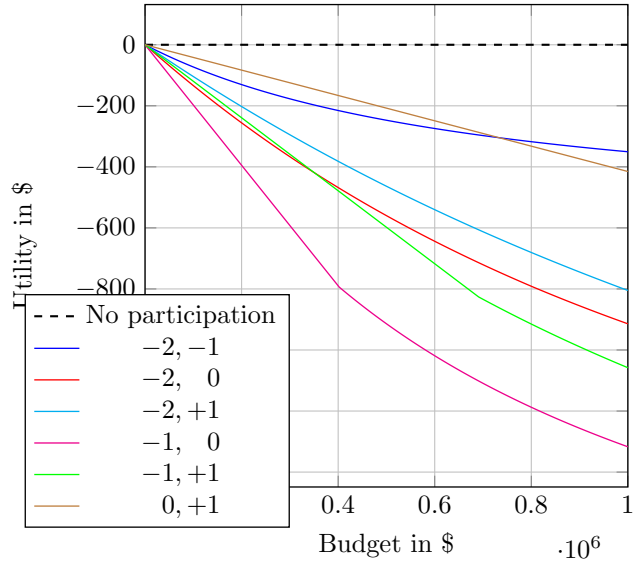

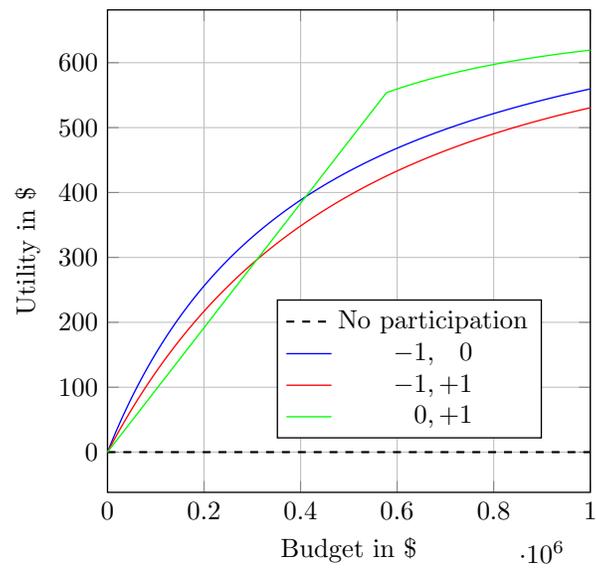
\begin{figure}
\centering
\begin{tikzpicture}
\begin{axis}[
    xlabel={Budget in \$},
    xmin=0,
    xmax=1000000,
    ylabel={Utility in \$},
    legend style={at={(0.9,0.4)}, anchor=north east},
    grid=both,
    width=8cm,
    height=8cm,
    line width=0.5pt
]
\addplot [domain=0:1000000, dashed, black, thick] {0};
\addlegendentry{No participation}

\addplot [mark=none, color=blue] table [x=budget, y={-1; 0}, col sep=comma] {data/utility_overprice_trade.csv};
\addlegendentry{$-1,  \phantom{-}0$}

\addplot [mark=none, color=red] table [x=budget, y={-1; +1}, col sep=comma] {data/utility_overprice_trade.csv};
\addlegendentry{$-1, +1$}

\addplot [mark=none, color=green] table [x=budget, y={0; +1}, col sep=comma] {data/utility_overprice_trade.csv};
\addlegendentry{$ \phantom{-}0, +1$}
\end{axis}
\end{tikzpicture}
\caption{\emph{Overpriced trade:} Utility curves for different choices of $(a,b)$ on an overpriced trade. The legend indicates how many ticks away the positions are relative to $\lfloor q \rfloor$. For example, the blue line labeled $-1, 0$ corresponds to the position $(t_{\lfloor q \rfloor - 1}, t_{\lfloor q \rfloor})$. All strategies yield positive utility, as the price impact consistently provides gains. The optimal strategy varies depending on the available budget.}
\label{fig:overprice}
\end{figure}

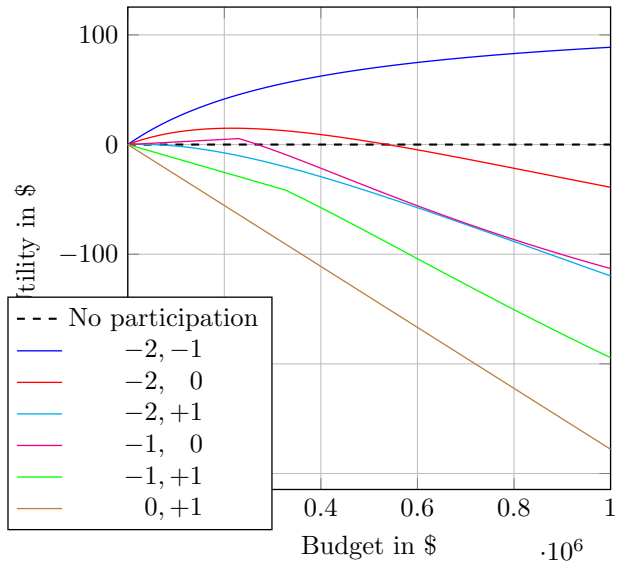
\begin{figure}
\centering
\begin{tikzpicture}
\begin{axis}[
    xlabel={Budget in \$},
    xmin=0,
    xmax=1000000,
    ylabel={Utility in \$},
    legend style={at={(0.3,0.4)}, anchor=north east},
    grid=both,
    width=8cm,
    height=8cm,
    line width=0.5pt
]

\addplot [domain=0:1000000, dashed, black, thick] {0};
\addlegendentry{No participation}

\addplot [mark=none, color=blue] table [x=budget, y={-2; -1}, col sep=comma] {data/utility_overshoot_trade.csv};
\addlegendentry{$-2, -1$}

\addplot [mark=none, color=red] table [x=budget, y={-2; 0}, col sep=comma] {data/utility_overshoot_trade.csv};
\addlegendentry{$-2,   \phantom{-}0$}

\addplot [mark=none, color=cyan] table [x=budget, y={-2; 1}, col sep=comma] {data/utility_overshoot_trade.csv};
\addlegendentry{$-2, +1$}

\addplot [mark=none, color=magenta] table [x=budget, y={-1; 0}, col sep=comma] {data/utility_overshoot_trade.csv};
\addlegendentry{$-1,  \phantom{-}0$}

\addplot [mark=none, color=green] table [x=budget, y={-1; +1}, col sep=comma] {data/utility_overshoot_trade.csv};
\addlegendentry{$-1, +1$}

\addplot [mark=none, color=brown] table [x=budget, y={0; +1}, col sep=comma] {data/utility_overshoot_trade.csv};
\addlegendentry{$ \phantom{-}0, +1$}

\end{axis}
\end{tikzpicture}
\caption{\emph{Overshoot trade:}  Utility curves for different choices of $(a,b)$ on an overshoot trade. The legend indicates how many ticks away the positions are relative to $\lfloor q \rfloor$. For example, the blue line labeled $-2, -1$ corresponds to the position $(t_{\lfloor q \rfloor - 2}, t_{\lfloor q \rfloor-1})$. Notice that some strategies achieve their optimal utility before exhausting the entire budget. The most profitable strategy concentrates budget on the position $(t_{\lfloor q \rfloor - 2}, t_{\lfloor q \rfloor - 1})$, which lies at the end of the trade, beyond the market price $p_x/p_y$, leading to a region where $\mathcal{C} < 0$.}
\label{fig:oversho}
\end{figure}

\end{document}